\documentclass[11pt,letterpaper]{article}
\usepackage[letterpaper, total={6.5in, 9in}]{geometry}

\usepackage{amsmath}
\usepackage{amssymb}
\usepackage{enumitem}
\usepackage{fancyhdr}
\usepackage{tikz}
\usepackage{bbm}
\usetikzlibrary{trees}
\usepackage{hyperref}

\usepackage{url}
\usepackage{graphicx}%
\usepackage{multirow}%
\usepackage{amsmath,amsfonts}%
\usepackage{mathrsfs}%
\usepackage[title]{appendix}%
\usepackage{xcolor}%
\usepackage{textcomp}%
\usepackage{manyfoot}%
\usepackage{booktabs}%
\usepackage{algorithm}%
\usepackage{algpseudocode}%
\usepackage{listings}%
\usepackage{chngcntr}
\usepackage{dsfont}
\usepackage{comment}

\newcommand{\R}{\mathbb{R}}

\newcommand{\norm}[1]{\left\lVert #1 \right\rVert}
\newcommand{\abs}[1]{\left\vert #1 \right\rvert}

\newcommand{\cov}[1]{\mathrm{Cov}{\left( #1\right)}}

\newcommand{\opnorm}[1]{\norm{#1}_{\operatorname{op}}}

\newcommand{\E}{\mathbb{E}}

\newcommand{\proj}{\operatorname{proj}}

\newcommand{\argmin}{\mathrm{argmin}\,}

\newcommand{\nn}{\nonumber}

\newcommand{\D}{\mathrm{d}}

\newcommand{\Var}{\mathsf{Var}}

\newcommand{\vol}{\mathsf{vol}}

\newcommand{\inner}[2]{\langle #1\,,#2 \rangle}

\newcommand{\brac}[1]{\left( #1\right)}

\newcommand*{\QEDA}{\hfill\hbox{\vrule width1.5ex height1.5ex}}

\newtheorem{thm}{Theorem}[section]
\newtheorem{theorem}[thm]{Theorem}

\newtheorem{lemma}[thm]{Lemma}

\newtheorem{remark}[thm]{Remark}

\usepackage{float}
\usepackage{caption}    
\usepackage{subcaption} 

\begin{document}
	\title{Oracle-based Uniform Sampling from Convex Bodies}
	\date{}
	\author{
    	Thanh Dang  \thanks{Department of Computer Science, University of Rochester, Rochester, NY 14620. 
        (email: {\tt ycloud777@gmail.com}).}
        \qquad
	 	Jiaming Liang \thanks{
        Goergen Institute for Data Science and Artificial Intelligence (GIDS-AI) and Department of Computer Science, University of Rochester, Rochester, NY 14620 (email: {\tt jiaming.liang@rochester.edu}). This work was partially supported by GIDS-AI seed funding.}
	}
\maketitle
	
	\begin{abstract}	
We propose new Markov chain Monte Carlo algorithms to sample a uniform distribution on a convex body $K$. Our algorithms are based on the proximal sampler, which uses Gibbs sampling on an augmented distribution and assumes access to the so-called restricted Gaussian oracle (RGO). The key contribution of this work is an efficient implementation of the RGO for uniform sampling on convex $K$ that goes beyond the membership-oracle model used in many classical and modern uniform samplers, and instead leverages richer oracle access commonly assumed in convex optimization. We implement the RGO via rejection sampling and access to either a projection oracle or a separation oracle on $K$. In both oracle models, we provide non-asymptotic complexity guarantees for obtaining unbiased samples, with accuracy quantified in R\'enyi divergence and $\chi^2$-divergence, and we support these theoretical guarantees with numerical experiments.


	\end{abstract}

\section{Introduction}
\label{sec:intro}

Sampling points from convex bodies in high dimension is a classical and central problem in computational geometry, probability, statistics, and optimization. Given a convex body $K\subseteq \R^d$, one likes to generate samples according to some distributions defined on $K$. Past and recent works in the area of constrained sampling in high dimension include \cite{lovasz1990mixing,cousins2018gaussian,kannan1997random,bubeck2018sampling,brosse2017sampling,kannan2009random,kook2024gaussian,girolami2011riemann,lee2018convergence,kook2023condition,gatmiry2023sampling,zhang2020wasserstein,jiang2021mirror,ahn2021efficient,li2022mirror,srinivasan2024fast,lehec2023langevin,gurbuzbalaban2024penalized,gopi2023algorithmic}; and many others.  In this paper, we will focus on the most fundamental case, i.e., uniform sampling on $K$. It is closely related to the problem of efficiently computing the volume of $K$, which is an important problem in computer science for the last few decades, see \cite{cousins2017efficient} and references therein. Moreover, uniform sampling also has a connection to Bayesian inference. If one takes the Gaussian distribution $\mathcal{N}(0,\sigma^2I_d)$, restricts it to $K$ and lets $\sigma$ becomes sufficiently large, then this truncated Gaussian distribution resembles the uniform distribution on $K$. At the same time, truncated Gaussian distribution has been used extensively in Bayesian statistical models with probit regression and censored data, see \cite{Held2006,albert1993bayesian,chib1992bayes,imai2005bayesian,tobin1958estimation}. Due to its importance, many works have been devoted to developing algorithms for this problem. The seminal work by \cite{dyer1991random} proposes the first algorithm to approximate the volume of any convex $K$ in polynomial time and also introduces the celebrated Ball walk to generate samples uniformly on $K$. 
Other algorithms that can be used for uniform sampling from convex $K$ are the Hit-and-Run walk in \cite{smith1984efficient,lovasz1993random,lovasz1999hit}, the coordinate Hit-and-Run walk \cite{turchin1971computation,diaconis2010gibbs,diaconis2012gibbs}, the Dikin walk in \cite{kannan2009random}, the Geodesic walk in \cite{lee2017geodesic}, and diffusion-based samplers in \cite{bubeck2018sampling,brosse2017sampling,gurbuzbalaban2024penalized,chalkis2023truncated}, among others. More details on these algorithms are provided in Appendix~\ref{appendix:uniformsampling}.

In \cite{leestructured21}, Lee, Shen, and Tian develop a Gibbs sampling scheme for log-concave sampling in high dimension. It is named Alternating Sampling Framework, which is commonly referred to as the proximal sampler. Each iteration of the scheme consists of an initial Gaussian step followed by a proximal-type sampling step. The proximal sampler is interesting to study in several ways. First, its Restricted Gaussian Oracle step is the sampling analogue of the proximal map in the proximal point method from optimization. Second, it is a high-accuracy/unbiased sampler (compared to Langevin Monte Carlo/Underdamped Langevin Monte Carlo, which are known to be biased). Finally, thanks to the discovery by \cite{chen2022improved}, the proximal sampler as a discrete-time Markov chain can be viewed through the lens of It\^{o} diffusion processes and analyzed using tools from stochastic calculus. For these reasons, the proximal sampler inspires many follow-up works such as \cite{gopi22a,chen2025rapid,pmlr-v272-mitra25a,wibisono2025mixing,yuan2023class,yuan2025proximal,liangchen2022proximal,fan2023improved,liang2023a} among others.


\textbf{Our contributions.} We develop efficient algorithms to perform uniform sampling on the convex body $K$, and we complement our theoretical guarantees with numerical experiments. Our algorithms are based on the proximal sampler. A central goal of our work is to incorporate more ideas from modern optimization into the design and analysis of sampling algorithms. In particular, we move beyond the classical membership-oracle model, which underlies many random-walk-based samplers including the In-and-Out algorithm by \cite{kook2024inandout}, and instead explore richer oracle structures such as \textbf{projection oracles} and \textbf{separation oracles}, which are standard oracles in convex optimization. More specifically, each outer iteration of the proximal sampler consists of an initial Gaussian step followed by a proximal-type sampling step. The latter is known as the Restricted Gaussian Oracle (RGO) and is the primary challenge in applying the proximal sampler. We propose Algorithm~\ref{alg:RGO:uniform:projection} and Algorithm~\ref{alg:RGO:uniform:separation} as implementations of the RGO via rejection sampling; they respectively require a \textbf{projection oracle} and a \textbf{separation oracle} on $K$. 

We summarize our main theoretical results as follows.
\begin{itemize}
    \item \textbf{Projection-oracle-based proximal sampler.}
    In Section~\ref{sec:projection_uniform}, we assume that $K$ satisfies $B(0,1)\subseteq K\subseteq B(0,R)$ where $B(0,R)$ is an $\ell_2$-ball centered at $0$ with radius $R \ge 1$, and that the initial distribution $\mu_0$ of the proximal sampler (Algorithm~\ref{alg:ASF_uniform}) satisfies a warm-start assumption: $d\mu_0/d\pi\leq M$ where $\pi$ is the uniform distribution on $K$. We implement the RGO via Algorithm~\ref{alg:RGO:uniform:projection}, which uses a \textbf{projection oracle} on $K$ and rejection sampling. Then, by Theorem~\ref{theo:outer} and Theorem~\ref{theo:averagerejection_projection}, to achieve $\epsilon$-accuracy in R\'{e}nyi divergence $\mathcal{R}_q$, Algorithm~\ref{alg:ASF_uniform} combined with Algorithm~\ref{alg:RGO:uniform:projection} requires at most
    \begin{align}
    \label{complexity}
       \mathcal{O}\brac{d^2{C_{\mathrm{LSI}}q\log \brac{2\frac{\log M}{\epsilon}} }}
    \end{align}
    proximal sampler iterations. Here $C_{\mathrm{LSI}}$, the LSI constant of the uniform distribution on $K$, is of order $\mathcal{O}(D^2)$ where $D$ is the diameter of $K$. Our result therefore matches the iteration complexity (in terms of dimension dependence and stepsize dependence) of In-and-Out \cite[Theorem~27]{kook2024inandout} and that of the Ball walk (see Appendix~\ref{appendix:uniformsampling}). Moreover, under the additional assumption that the stepsize $\eta$ of the proximal sampler equals $1/d^2$, Algorithm~\ref{alg:RGO:uniform:projection} and Theorem~\ref{theo:averagerejection_projection} also imply that each iteration of the proximal sampler makes one query to the \textbf{projection oracle} on $K$ and has at most an average of $M(\sqrt{2\pi e}+1)$ proposals for the rejection sampling. Finally, Theorem~\ref{theo:outer} and Theorem~\ref{theo:averagerejection_projection} also provide analogous results in $\chi^2$-divergence.

    \item \textbf{Separation-oracle-based proximal sampler.}
    In Section~\ref{sec:separation_uniform},  we implement the RGO via Algorithm~\ref{alg:RGO:uniform:separation}, which uses a \textbf{separation oracle} on $K$ and rejection sampling. We assume the same conditions on $K$ and $\mu_0$ as above. Then Theorem~\ref{theo:outer} and Theorem~\ref{theo:averagerejection_separation} imply that the number of iterations of the proximal sampler (Algorithm~\ref{alg:ASF_uniform}) to reach $\epsilon$-accuracy in R\'{e}nyi divergence $\mathcal{R}_q$ has the same bound as in \eqref{complexity}. In particular, each iteration of the proximal sampler makes ${\cal O}\!\left(d \log \frac{d \gamma}{\alpha}\right)$ queries to the \textbf{separation oracle} on $K$, where $\gamma=R/\mbox{minwidth}(K)$, $\mbox{minwidth}(K)=\min_{\norm{a}=1}\brac{\max_{y\in K}a^Ty-\min_{y\in K}a^Ty}$, and the constant $\alpha\in(0,1)$ is not too small, i.e., it satisfies $\Pr\brac{\alpha\leq \frac{2}{7 d^3 R^2} }\leq 4 \exp\brac{-\frac{d^2R^2}{8}}$.
    In addition, the average number of proposals is $M\sqrt{2\pi} \exp\brac{\frac{13}{4}+\frac{20}{d}} + M \exp\brac{\frac{9}{4}+\frac{12}{d}}$ for the rejection sampling. Finally, Theorem~\ref{theo:outer} and Theorem~\ref{theo:averagerejection_separation} also provide analogous results in $\chi^2$-divergence.
\end{itemize}

Regarding the numerical experiments presented in Section~\ref{section:numerical}, we compare  our PS+Alg3 (Algorithm~\ref{alg:ASF_uniform}+Algorithm~\ref{alg:RGO:uniform:projection}) with In-and-Out (Algorithm~1 in \cite{kook2024inandout}) on random dense $Z$-polytopes, which are standard examples of convex bodies and are defined in Appendix~\ref{appen_examplesconvexbodies}. We consider two stepsizes: the conservative choice $\eta_A=10^{-5}$ (as recommended by our Theorem~\ref{theo:averagerejection_projection}) and an aggressive practical choice $\eta_B=10^{-2}$. As an accuracy proxy, we report the total variation (TV) error of the $x_1$ marginal against a long hit-and-run reference chain, and we also report elapsed time to capture solver/oracle overhead. Under the conservative stepsize, PS+Alg3 reduces TV faster and reaches a substantially lower final TV than In-and-Out, yielding a better accuracy-time tradeoff despite higher per-iteration cost. Under the aggressive stepsize, In-and-Out remains stable and achieves a final TV comparable to PS+Alg3, while PS+Alg3 reaches a given TV level substantially sooner in elapsed time.

\paragraph*{Comparison to prior works.}
 To the best of our knowledge, there are two prior works on constrained sampling that design samplers using a \textbf{projection oracle} for the constraint set, namely \cite{Bubeck2018,Lehec2023}. In particular, \cite{Lehec2023} studies sampling from a density proportional to $e^{-f}$ over a convex set $K\subset\mathbb{R}^d$, where $f$ is convex and Lipschitz continuous. They assume access to a \textbf{projection oracle} over $K$ and derive an iteration complexity bound of $O\!\left(\frac{d}{\epsilon^{2}}\right)$, where $d$ is the dimension and $\epsilon$ is the target accuracy in Wasserstein distance. Comparing the iteration complexity of \cite{Lehec2023} to \eqref{complexity}, our guarantee is stated in a stronger metric (R\'enyi divergence and $\chi^{2}$ divergence), has better $\epsilon$-dependence but worse dependence on the dimension $d$.

Meanwhile, to the best of our knowledge, there is one earlier work by \cite{kannan1997random} that designs a constrained sampler under a \textbf{separation-oracle} model. The uniform-sampling algorithm therein has oracle complexity $O\!\left(d^{3} \log^{2}\!\frac{d}{\epsilon}\right)$ with respect to the total variation distance \cite[Theorems~2.2 and~4.14]{kannan1997random}. In contrast, our proximal sampler with the RGO implementation given by Algorithm~\ref{alg:RGO:uniform:separation} achieves a comparable complexity $O\!\left(d^{3}\log d \cdot \log\frac{1}{\epsilon}\right)$ in stronger metrics, namely R\'enyi divergence and $\chi^{2}$ divergence.

It is also important to mention the recent work by \cite{kook2024inandout}. Therein, Kook, Vempala, and Zhang  propose the In-and-Out algorithm for uniform sampling from convex $K$. Their algorithm is based on the proximal sampler and they employ a clever smoothing argument to adapt the proof technique by \cite{chen2022improved} to the case of the uniform distribution from $K$. In particular, their implementation of the RGO step is via a \textbf{membership oracle}: they sample $x$ from a Gaussian distribution until $x\in K$ up to a certain number of maximum attempts $N$, at which point the algorithm halts and declares failure. Then they carefully analyze the condition on stepsize and $N$ to make sure the failure probability is small. Subsequent works using the proximal sampler for uniform sampling on $K$ and more general log-concave sampling by reducing it to the problem of exponential sampling from convex bodies have been carried out in \cite{kook2025algodiffusion,kook2025coldstart,kook2025zeroth}. In particular, \cite{kook2025zeroth} is able to get rid of the failure probability in the In-and-Out algorithm by introducing a \textit{restart} step: if one cannot generate $x$ that is in $K$ from a Gaussian distribution after $N$ attempts in the second step of the proximal sampler, then \textit{restart} by returning to the first step.  

Finally, we highlight two additional features of our approach. First, our RGO implementations are unbiased and their outputs lie in $K$ almost surely (see Remark~\ref{remark:feasible}). Thus we avoid the failure probability in \cite{kook2024inandout,kook2025algodiffusion,kook2025coldstart} and provide an alternative to the \textit{restart} procedure of \cite{kook2025zeroth}. Second, to facilitate implementation, Appendix~\ref{appen_examplesconvexbodies} details how to realize \textbf{membership}, \textbf{projection}, and \textbf{separation oracles} for common convex bodies. It turns out in many standard examples, these oracles have comparable computational cost to implement.

\section{Preliminaries}
\label{sec:prelim}

\subsection{Notation, definitions, and assumptions}
Regarding notation, $\norm{\cdot}$ denotes the Euclidean norm on $\R^d$, $\opnorm{\cdot}$ the operator norm, and $I_d$ the $d\times d$ identity. The notation $x=\mathcal{O}(a)$ means $x\le Ca$ for a universal constant $C>0$, and $\tilde{\mathcal{O}}(a)$ allows additional logarithmic factors.

\textbf{Absolute continuity.} For measures $\mu,\nu$ on $(E,\mathcal{F})$, we write $\mu\ll\nu$ if there exists $f:E\to\R$ such that $\mu(A)=\int_A f d\nu$ for all $A\in\mathcal{F}$. The function $f$ is the Radon-Nikodym derivative, denoted $\frac{\D\mu}{\D\nu}$.

\textbf{Metric.} Let $\phi:\R_{\ge0}\to\R$ be convex with $\phi(1)=0$. For probability measures $\mu\ll\nu$ on $(E,\mathcal{F})$, define the $\phi$-divergence $D_\phi(\mu||\nu)=\int_E \phi\!\left(\frac{\D\mu}{\D\nu}\right)\D\nu$. For $\phi(x)=x\log x$ this is the Kullback-Leibler divergence, and for $\phi(x)=x^2-1$ it is the $\chi^2$-divergence. For $q>0$, the $q$-R\'enyi divergence is $\mathcal{R}_q(\mu||\nu)=\frac{1}{q-1}\log\!\bigl(\chi^q(\mu||\nu)+1\bigr)$. The relative Fisher information is $\mathrm{FI}(\mu||\nu)=\int_E \norm{\nabla \log \frac{\D\mu}{\D\nu}}^2 d\mu$.

\textbf{Volumes, distance to a set and $\mbox{minwidth}$.} Let $\vol(K)$ and $\vol_{d-1}(\partial K)$ respectively denote the volumes of $K\subseteq \R^d$ and the boundary set $\partial K \subseteq \R^{d-1}$. Let $d (y,K):=\inf_{z\in S}\|y-z\|$ denote the Euclidean distance from $y$ to $S$. Set $\mbox{minwidth}(K)=\min_{\norm{a}=1}\brac{\max_{y\in K}a^Ty-\min_{y\in K}a^Ty}$.

\textbf{Normalizing constants.} 
For a measurable $\Theta:\R^{m}\to\R\cup\{+\infty\}$, define the normalizing constant $N_\Theta:=\int_{\R^m} e^{-\Theta(z)}dz$ whenever this integral is finite.

\textbf{Isoperimetric inequalities and isoperimetric constants.} We say $\nu$ satisfies the log-Sobolev inequality (LSI) with constant $C_{\mathrm{LSI}}$ if for all $\mu\ll\nu$, $\mathrm{KL}(\mu||\nu)\le \frac{C_{\mathrm{LSI}}}{2} \mathrm{FI}(\mu||\nu)$. We say $\nu$ satisfies the Poincar\'e inequality (PI) with constant $C_{\mathrm{PI}}$ if for any smooth bounded $\psi$, $\Var_{\nu}(\psi)\le C_{\mathrm{PI}} \mathbb{E}_\nu\!\left[\norm{\nabla \psi}^2\right]$.

In \cite[Appendix~C]{kook2024inandout}, the authors summarize recent bounds on isoperimetric constants \cite{cheeger2015lower,kannan1995isoperimetric,lee2024eldan,chen2021almost,klartag2022bourgain,klartag2023logarithmic}. As a consequence, we obtain the following bounds on the LSI and PI constants for the uniform law on a convex body $K$. Recall that a probability measure $\pi$ on $K$ is \textit{isotropic} if, for $X\sim\pi$, we have $\E X_i=0$ and $\E[X_iX_j]=\mathbf{1}_{i=j}$ for all $i,j$.

    \begin{lemma}
    \label{lem:isoconstant}
    \cite[Lemma~18]{kook2024inandout}
    Let $\pi$ be the uniform distribution over $K$ and $K\subset \R^d$ be a convex body with diameter $D$, where $D=\max_{x,y\in K}\norm{x-y}$. Then we have   $C_{\mathrm{PI}}(\pi)=  \mathcal{O}\brac{\opnorm{\cov {\pi}}\log d}$ and $C_\mathrm{LSI}(\pi)=\mathcal{O} \brac{D^2}$. In particular, if $\pi$ is isotropic, then $C_{\mathrm{PI}}(\pi)=  \mathcal{O}\brac{\log d}$ and $C_\mathrm{LSI}(\pi)=\mathcal{O} \brac{D}$.
\end{lemma}

\textbf{Oracles.} A membership oracle decides whether a query $x\in\R^d$ lies in $K$. A separation oracle either certifies $x\in K$ or, if $x\notin K$, returns $g:\R^d\to\R^d$ such that $\inner{g(x)}{x-y}\ge 0$ for all $y\in K$; in particular, it subsumes membership. A projection oracle returns $\proj_K(y)=\argmin\{\norm{x-y}^2: x\in K\}$; clearly $\proj_K(y)=y$ if $y\in K$.

\textbf{Warmness.} For probability measures $\mu,\nu$ on $\R^d$ and $M>0$, we say $\mu$ is $M$-warm w.r.t.\ $\nu$ if $\mu\ll\nu$ and $\frac{\D\mu}{\D\nu}(x)\le M$ for all $x\in\R^d$. In Algorithms~\ref{alg:RGO:uniform:projection} and~\ref{alg:RGO:uniform:separation}, we assume a warm start: $\mu_0$ is $M$-warm w.r.t.\ the uniform distribution $\pi$ on $K$.

\textbf{Standing assumptions for the paper.} In both Sections \ref{sec:projection_uniform} and \ref{sec:separation_uniform}, we assume the following conditions hold:
\begin{itemize}
\item[(A1)] $K$ is a non-empty, closed, and convex set in $\R^d$ such that $B(0,1) \subseteq K \subseteq B(0,R)$ for some $R>1$, where $B(0,R)$ denotes the Euclidean ball centered at the origin with radius $R$. 
    \item[(A2)] the initial distribution $\mu_0$ is $M$-warm with respect to the uniform distribution $\pi$ on $K$, i.e., $\frac{\D\mu_0}{\D \pi}\leq M$.
\end{itemize}

\subsection{The proximal sampler}
The proximal sampler is first proposed by \cite{leestructured21} to sample log-concave distribution in $\R^d$. Given a stepsize $\eta>0$, it aims to sample the target distribution $\pi^X(x)\sim \exp(-f(x))$ by performing Gibbs sampling for the augmented distribution $\pi^{X,Y}(x,y)\sim \exp\brac{-f(x)-\norm{x-y}^2/(2\eta)}$ whose $X$-marginal is the target $\pi^X$. This idea of sampling from a joint distribution to obtain the marginal distribution has been observed in earlier references, for example \cite{cousins2018gaussian}. Each iteration of the proximal sampler alternates between two steps:
\begin{algorithm}[H]
	\caption{Proximal sampler by \cite{leestructured21}}
	\label{alg:ASF}
	\begin{algorithmic}
		\State 1. Sample $y_k\sim \pi^{Y|X}(y|x_k)\propto \exp(-\frac{1}{2\eta}\|x_k-y\|^2)$;
		\State 2. Sample $x_{k+1}\sim  \pi^{X \mid Y}(x \mid y_k)\propto \exp(-f(x)-\frac{1}{2\eta}\|x-y_k\|^2)$.
	\end{algorithmic}
\end{algorithm}
While \cite{leestructured21} proposes the proximal sampler for log-concave sampling, \cite{chen2022improved} extends the assumption of log-concave distributions to distributions satisfying common isoperimetric inequalities such as log-Sobolev inequality or Poincar\'{e} inequality. The crucial observation by \cite{chen2022improved} is that while the proximal sampler is a Markov chain, each iteration of this chain can be viewed as a pair of forward and backward diffusion steps where probabilistic tools for  It\^{o} diffusion processes can be applied. For the proximal sampler, i.e., Algorithm~\ref{alg:ASF}, the first step is generating a Gaussian sample and thus can be easily done, while the second step is non-trivial RGO. In both \cite{leestructured21,chen2022improved}, the authors either assume they have exact access to the RGO, or that $f$ is smooth so that RGO can be easily realized via rejection sampling. Novel realizations of the RGO to either reduce its cost or to relax the smoothness assumption have been investigated in \cite{liangchen2022proximal,gopi22a,fan2023improved,yuan2023class,liang2023a,liang2024proximal} among others.

The uniform distribution on $K$ has density proportional to $\mathbf{1}_K$, where $ \mathbf{1}_K(x)$ equals $1$ on $K$ and $0$ otherwise. In the context of the proximal sampler introduced above, if we take $f(x)=  I_K(x)$ where the indicator function $I_K(x)$ equals $0$ if $x\in K$ and equals $+\infty$ otherwise, then the target of the proximal sampler will be the uniform distribution on $K$. In particular, we have 
       \begin{equation}\label{eq:piXY}
           \pi^X\propto \exp(-I_K(x))=\mathbf{1}_K(x), \qquad \pi^{X,Y}(x,y)\propto \exp\brac{-\frac{1}{2\eta} \norm{x-y}^2}\mathbf{1}_K(x). 
       \end{equation}
 Moreover, denote ${\cal N}(y,\eta I_d)|_K $ the Gaussian distribution ${\cal N}(y,\eta I_d)$ restricted to $K$, i.e.,
\begin{equation}\label{eq:truncated}
  {\cal N}(y,\eta I_d)|_K \propto \exp\brac{-\frac{1}{2\eta}\norm{x-y}^2-I_K(x)}=  \exp\brac{-\frac{1}{2\eta}\norm{x-y}^2}\mathbf{1}_K(x). 
\end{equation}
Then, Algorithm \ref{alg:ASF} for uniform sampling on $K$ turns into the following proximal sampler.
\begin{algorithm}[H]
	\caption{Proximal sampler for the uniform distribution on $K$}
	\label{alg:ASF_uniform}
	\begin{algorithmic}
		\State 1. Generate $y_k\sim \pi^{Y|X}(y|x_k)= \mathcal{N}\brac{x_k,\eta I_d}$;
		\State 2. Generate $x_{k+1}\sim \pi^{X \mid Y}(x \mid y_k)={\cal N}(y_k,\eta I_d)|_K. $
	\end{algorithmic}
\end{algorithm}

The next theorem follows directly from the one-step contraction in $\chi^2$ and R\'enyi divergences proved in~\cite[Theorem 23]{kook2025algodiffusion} (restated as Theorem~\ref{theo:kooketal} in Appendix~\ref{appendix:techlem}). Assuming the RGO step (Algorithm~\ref{alg:ASF_uniform}, Step 2) is available and exact, we obtain the iteration complexity of the proximal sampler in the next theorem. The proof is in Appendix~\ref{appendix:techlem}. 


\begin{theorem}
    \label{theo:outer}
 Let $K\subset \R^d$ be a convex set. Assume Algorithm \ref{alg:ASF_uniform} starts from an $M$-warm distribution $\mu^X_0$, i.e., $\frac{\D\mu_0}{\D \pi^X}\leq M$. Let $\epsilon>0$. Denote $C_{\mathrm{PI}}$ and $C_{\mathrm{LSI}}$ respectively the Poincar\'{e} constant and the log Sobolev constant of the uniform distribution $\pi^X$ on $K$. Then,

\begin{enumerate}[label=\alph*)]
    \item with respect to the R\'{e}nyi divergence $\mathcal{R}_q$ and $q\geq 1$,  the algorithm can achieve $\epsilon$-accuracy within $ O\brac{\frac{1}{\eta}{C_{\mathrm{LSI}}q\log \brac{2\frac{\log M}{\epsilon}} }}$
iterations, where, per Lemma~\ref{lem:isoconstant}, $C_{\mathrm{LSI}}=\mathcal{O}(D^2)$ in general and $C_{\mathrm{LSI}}=\mathcal{O}(D)$ if $\pi^X$ is isotropic;

\item with respect to the $\chi^2$-divergence, the algorithm reaches $\epsilon$-accuracy in $O\brac{\frac{1}{\eta}{C_{\mathrm{PI}}\log \brac{2\frac{M^2+1}{\epsilon}}  } } $
iterations, where, per Lemma~\ref{lem:isoconstant}, $C_{\mathrm{PI}}(\pi)=  \mathcal{O}\brac{\opnorm{\cov {\pi}}\log d}$ in general and $C_{\mathrm{PI}}(\pi)=  \mathcal{O}\brac{\log d}$ if $\pi^X$ is isotropic. 
\end{enumerate}

\end{theorem}
The iteration complexities provided in Theorem~\ref{theo:outer} have not taken into account methods to implement the RGO (Step 2 in Algorithm \ref{alg:ASF_uniform}) and the costs associated with them. Our method for the RGO implementation in the upcoming sections is based on rejection sampling, which requires the construction of a suitable sampling proposal that is close to the target $\pi^{X|Y}(x|y)={\cal N}(y,\eta I_d)|_K$ in Algorithm \ref{alg:ASF_uniform}. We will see that one can construct quite natural proposals if given access to either a projection oracle or a separation oracle on $K$.

\section{Projection oracle-based proximal sampling} \label{sec:projection_uniform}

This section aims to implement the RGO step, i.e., Step 2 in Algorithm \ref{alg:ASF_uniform}, via rejection sampling and the projection oracle of $K$.
At the $k$-th iteration, Step 1 in Algorithm \ref{alg:ASF_uniform} generates from ${\cal N}(x_k,\eta I_d)$ a point $y:=y_k$, which is fixed in RGO. Then the RGO step is supposed to sample from the truncated Gaussian ${\cal N}(y,\eta I_d)|_K$. To implement RGO by rejection sampling, we need to construct a proposal that is both easier to sample than ${\cal N}(y,\eta I_d)|_K$ and also reasonably close to ${\cal N}(y,\eta I_d)|_K$ in order to ensure the acceptance probability is high, or equivalently, the number of proposals is low.
Examining the equivalent formulas of ${\cal N}(y,\eta I_d)|_K$ in \eqref{eq:truncated}, one can easily figure out that $\mathcal{N}(y,\eta I_d)|_K$ concentrates at
\begin{equation} \label{regularizedmap}
    \underset{x\in \R^d}\argmin \left\{\Theta_y^{\eta,K}(x) := I_K(x) + \frac{1}{2\eta}\|x-y\|^2\right\},
\end{equation}
which is precisely the projection of $y$ onto $K$, i.e., $\proj_K(y)$.  
Inspired by this observation, the proposal we choose for the rejection sampling is the Gaussian distribution ${\cal N}(\proj_K(y),\eta I_d)$.

Below is our implementation of RGO via the projection oracle $\proj_K$ and rejection sampling. Let ${\cal U}[0,1]$ denote the uniform distribution on $[0,1]$. 

\begin{algorithm}[H]
	\caption{Projection oracle-based implementation of RGO} 
	\label{alg:RGO:uniform:projection}
	\begin{algorithmic}
		\State 1. Generate  $X\sim \mathcal{N}\brac{\proj_K(y),\eta I_d}$ and $U\sim {\cal U}[0,1]$. 
		\State 2. If 
		\begin{equation}\label{eq:event_uniform}
			U\leq \exp \brac{-\frac{1}{\eta}\inner{X-\proj_K(y)}{\proj_K(y)-y} } \mathbf{1}_K(X), 
		\end{equation}
		then accept $X$; otherwise, reject $X$ and go to step 1. 
	\end{algorithmic}
\end{algorithm}

\begin{remark}
\label{remark:feasible} 
By Lemma~\ref{lem:rejection} (Appendix~\ref{appendix:techlem}), the output $X$ of Algorithm~\ref{alg:RGO:uniform:projection} has law $\pi^{X|Y}$, so the RGO implementation is unbiased. Moreover, the acceptance test \eqref{eq:event_uniform} enforces feasibility: if $X\notin K$ then $\mathbf{1}_K(X)=0$, hence $\Pr(\text{$X$ accepted and }X\notin K)=\Pr(U\le 0)=0$. Section~\ref{sec:separation_uniform} gives an alternative RGO implementation via a separation oracle that also accepts only feasible points.

In contrast, \cite{kook2024inandout} draws $x_i\sim{\cal N}(y,\eta I_d)$ up to $\tilde{\cal O}(d^2)$ times and accepts the first $x_i\in K$, declaring failure otherwise \cite[Remark~2]{kook2024inandout}. This membership-oracle approach has nonzero failure probability, noting that this failure is removed by \cite{kook2025zeroth} via a \textit{restart} procedure.  Here, we are offering an alternative to the \textit{restart} procedure by using a projection oracle or a separation oracle on $K$.

\end{remark}

In the next lemma (proved in Appendix~\ref{appendix:techlem}), we verify that the acceptance test \eqref{eq:event_uniform} is well-defined. We also introduce a function $\mathcal{P}_1$, which will appear naturally in later rejection-sampling analysis.

\begin{lemma}\label{lem:compareP1} 
Assume condition (A1) holds. Then for every $x\in \R^d$, we have 
\[ -I_K(x)-\frac{1}{\eta}\inner{x-\proj_K(y)}{\proj_K(y)-y} \le 0\]
and hence the acceptance test \eqref{eq:event_uniform} is well-defined.
Moreover, \eqref{eq:event_uniform} is equivalent to  $U \le \exp(\mathcal{P}_1(x)- \Theta_y^{\eta,K}(x))$
where $\Theta_y^{\eta,K}$ is as in \eqref{regularizedmap} and 
\begin{equation}\label{def:P1}
    \mathcal{P}_1(x)
    = \frac{1}{2\eta}\|x-\proj_K(y)\|^2 + \frac{1}{2\eta}\|\proj_K(y)-y\|^2. 
\end{equation}
\end{lemma}

The following lemma is a key technical contribution. It enables the bounds on the average number of proposals of Algorithm \ref{alg:RGO:uniform:projection} in Theorem~\ref{theo:averagerejection_projection} and of Algorithm \ref{alg:RGO:uniform:separation} in Theorem~\ref{theo:averagerejection_separation}.

\begin{lemma}\label{lem:intR}
Let $\tau\ge 0$ be given and assume condition (A1) holds, then we have
\[
      \int_{\R^d} \exp\brac{-\frac{\brac{\|\proj_K(y)-y\|-\tau}^2}{2\eta}}\D y 
    \le \vol(K) \left[\exp\left(\frac{\eta d^2}{2}+\tau d\right) \sqrt{2\pi \eta d^2}+\exp\brac{-\frac{\tau^2}{2\eta}}\right]. 
\]
\end{lemma}

\begin{proof}
Since $\proj_K(y)=y$ for $y\in K$, it follows that
\begin{equation}\label{ineq:intK}
    \int_K \exp\brac{-\frac{\brac{\|\proj_K(y)-y\|-\tau}^2}{2\eta}}\D y 
    = \vol(K)\exp\brac{-\frac{\tau^2}{2\eta}}.
\end{equation}
    Next, let us set $K_\delta=\{x\in \R^d:d(x,K)\leq \delta\}$ where $d(x,K)$ denotes the distance from $x$ to $K$. Then, by the co-area formula, we can write
	\begin{align}
		\int_{K^c} \exp\left(-\frac{\brac{\|\proj_K(y)-y\|-\tau}^2}{2\eta}\right) \D y & = \int_{K^c} \exp\left(-\frac{\brac{d(y,K)-\tau}^2}{2\eta}\right) \D y \nn \\
		& = \int_0^\infty \exp\left(-\frac{\brac{\delta-\tau}^2}{2\eta}\right) \vol_{d-1}(\partial K_\delta) \D \delta. \label{restrictedgaussianintegral}
	\end{align}
	It follows from $B(0,1) \subseteq K$ in condition (A1) that $K_\delta = K+\delta B(0,1) \subseteq (1+\delta) K$. This relation, the fact that $(1+\delta)^d \le \exp(\delta d)$, and Lemma~\ref{lem:vol} with $K=K_\delta$ together imply that
    \begin{align}
    \label{ineq:volpartialKdelta}
    \vol_{d-1}(\partial K_\delta) \stackrel{\eqref{ineq:vol}}\le d\vol(K_\delta) \le (1+\delta)^d d \vol(K) \le e^{\delta d} d \vol(K).
    \end{align}
Plugging this inequality~\eqref{ineq:volpartialKdelta} into \eqref{restrictedgaussianintegral}, we obtain
	\begin{align}
    \label{boundKc}
		&\int_{K^c} \exp\left(-\frac{\brac{\|\proj_K(y)-y\|-\tau}^2}{2\eta}\right) \D y 
        \stackrel{\eqref{restrictedgaussianintegral},\eqref{ineq:volpartialKdelta}}\le d \vol(K) \int_0^\infty \exp\left(-\frac{\brac{\delta-\tau}^2}{2\eta} + \delta d\right) \D \delta \nonumber \\
		=& d \vol(K) \exp\left(\frac{\eta d^2}{2}+\tau d\right)\int_0^\infty \exp\left(-\frac{1}{2\eta} (\delta-\tau - \eta d)^2 \right) \D \delta \nonumber\\
		=& d \vol(K) \exp\left(\frac{\eta d^2}{2}+\tau d\right)\int_{-\eta d-\tau}^\infty \exp\left(-\frac{a^2}{2\eta}  \right) \D a 
		\le \vol(K) \exp\left(\frac{\eta d^2}{2}+\tau d\right) \sqrt{2\pi \eta d^2}, 
	\end{align}
where $a=\delta -\tau- \eta d$. The stated lemma finally follows from combining \eqref{boundKc} and \eqref{ineq:intK}.
\end{proof}

\begin{remark}
\label{rem:alternative}
    Regarding the proof argument above, we provide in Lemma~\ref{lem:alternativeintbypart} an alternative way of bounding the terms in \eqref{restrictedgaussianintegral}, which uses integration by parts and the fact that $\frac{\D}{\D\delta} \vol(K_\delta)=\vol_{d-1}(\partial K_\delta)$ for $\delta\geq 0$ a.s.. 
\end{remark}

Next, we bound the average number of proposals in Algorithm~\ref{alg:RGO:uniform:projection} per each iteration of Algorithm~\ref{alg:ASF_uniform}. With the stepsize $\eta=1/d^2$, this average is $\mathcal{O}(1)$ and matches the dimension scaling in \cite[Theorem~27]{kook2024inandout}.

\begin{theorem}
\label{theo:averagerejection_projection} Assume conditions (A1) and (A2) hold, and consider Algorithm \ref{alg:ASF_uniform} with stepsize $\eta =1/d^2$.  Then, the average number of proposals in Algorithm \ref{alg:RGO:uniform:projection} is bounded by $M(\sqrt{2\pi e}+1)$.
\end{theorem}

\begin{proof}
 Denote $\mu^k$ the distribution of $y=y_k$ for the first step of Algorithm \ref{alg:ASF_uniform}. Write $n_y$ the average number of proposals conditioned on $y$ in the rejection sampler that is Algorithm~\ref{alg:RGO:uniform:projection}. Per Lemma~\ref{lem:rejection}, the average number of proposals is $\E_{\mu^k}[n_y]$ where the formula of $n_y$ is given in \eqref{eq:generic-accept-rate}. The fact that $d\mu^k/d\pi^Y\leq M$ from Lemma~\ref{lem:warmstart} implies $\E_{\mu^k}[n_y] \le M \E_{\pi^Y}[n_y]$, 
and hence we will focus on bounding $\E_{\pi^Y}[n_y]$. In view of \eqref{eq:generic-accept-rate}, the latter expression becomes
\begin{align*}
    \E_{\pi^Y}[n_y]  = \E_{\pi^Y}\left[\frac{N_{\mathcal{P}_1}}{N_{\Theta_y^{\eta,K}} }\right]= \int_{\R^d} \frac{\int_{\R^d} \exp(-\mathcal{P}_1(x)) \D x}{\int_K \exp\left(-\frac{1}{2\eta}\|x-y\|^2\right) \D x} \pi^Y(y) \D y. 
\end{align*}
 Via Lemma~\ref{lem:gaussianint}(a), it is easy to get  
 \[
 \int_{\R^d} \exp(-\mathcal{P}_1(x)) \D x = (2\pi \eta)^{d/2} \exp\left(-\frac{1}{2\eta}\|\proj_K(y)-y\|^2\right). 
 \]
 Furthermore, via \eqref{eq:piXY} and Lemma~\ref{lem:gaussianint}(a), one obtains
\begin{align}
\label{for:piy}
          \pi^Y(y) = \frac{\int_{\R^d} \pi^{X,Y}(x,y) \D x}{\int_{\R^d} \int_{\R^d} \pi^{X,Y}(x,y) \D x \D y} \stackrel{\eqref{eq:piXY}}= \frac{1}{\vol(K) (2\pi \eta)^{d/2}} \int_K \exp\left(-\frac{1}{2\eta}\|x-y\|^2\right) \D x. 
\end{align}
The last three identities yield $	\E_{\pi^Y}[n_y] 
	= \frac{1}{\vol(K)}  \int_{\R^d} \exp\left(-\frac{1}{2\eta}\|\proj_K(y)-y\|^2\right) \D y$. 
	Finally, it follows from Lemma~\ref{lem:intR} with $\tau =0$ that
	\[
	\E_{\pi^Y}[n_y] \le \exp\left(\frac{\eta d^2}{2}\right) \sqrt{2\pi \eta d^2} +1\quad\text{and}\quad \E_{\mu^k}[n_y] \le M\exp\left(\frac{\eta d^2}{2}\right) \sqrt{2\pi \eta d^2}+M.
	\]
	The conclusion immediately follows by taking $\eta = 1/d^2$. This completes the proof. 
\end{proof}


\section{Separation oracle-based proximal sampling}

\label{sec:separation_uniform}

In Section \ref{sec:projection_uniform}, when a projection oracle is available, one get an exact solution to $\argmin_{x\in \R^d} \Theta^{\eta,K}_y(x)$ as $\proj_K(y)$, and the Gaussian proposal for the rejection sampling can thus be centered at $\proj_K(y)$. When $\proj_K$ is not available, we propose to use a state-of-the-art Cutting Plane method by \cite{jiang2020cuttingplane}, which uses a separation oracle on $K$ to find an approximate solution $\hat{x}$ of $\argmin_{x\in \R^d} \Theta^{\eta,K}_y(x)$. At this point, using $\hat{x}$, we do rejection sampling with the proposal
\[
\nu(x)\ \propto\ \exp\left(-\frac{1}{2\eta}\left[\|x-\hat x\|^2-2\sqrt{\frac{2\eta}{d}}\|x-\hat x\|\right]\right).
\]
As in Section~\ref{sec:projection_uniform}, this proposal is centered at a high-concentration point $\hat x$, but it is not Gaussian. Sampling from $\nu$ is nevertheless simple: it reduces to a one-dimensional sampling problem (see Lemma~\ref{lem:proposalsampling} in Appendix~\ref{appendix:secseparation}).

Below are our RGO implementation assuming access to a separation oracle on $K$.
\begin{algorithm}[H]
	\caption{Separation oracle-based implementation of RGO} 
	\label{alg:RGO:uniform:separation}
	\begin{algorithmic}
		\State 1. Compute a $(1/d)$-solution $\hat{x}$ of $\argmin_{x\in \R^d} \Theta^{\eta,K}_y(x)$ via the Cutting Plane method by \cite{jiang2020cuttingplane}.
		\State 2. Generate $U\sim {\cal U}[0,1]$ and $X\sim \nu(x)$ via Algorithm \ref{alg:separation:Psamplingforseparation} in Appendix \ref{appendix:generatesampleseparation}.
		\State 3. If 
		\begin{equation}\label{eq:event_uniform:separation}
			U \le \exp\brac{\mathcal{P}_2(X) - \Theta_y^{\eta,K}(X)},
		\end{equation}
		then accept $X$; otherwise, reject $X$ and go to step 2. The function $\mathcal{P}_2$ is defined in \eqref{def:P2} below.
	\end{algorithmic}
\end{algorithm}

Lemma~\ref{lem:cuttingplanealg:deltasolution} in Appendix~\ref{appendix:cuttingplane} bounds the number of separation-oracle calls needed to produce $\hat x$ in Step~1. Lemma~\ref{lem:rejection} in Appendix~\ref{appendix:techlem} shows that the rejection sampling is unbiased.

The following result (proved in Appendix~\ref{proof:alltheP}) ensure the acceptance test at \eqref{eq:event_uniform:separation} is well-defined. 

\begin{lemma} \label{lem:comparealltheP}
Assume condition (A1) holds. Recall $\Theta^{\eta,K}_y$ and $\mathcal{P}_1$ defined in \eqref{regularizedmap} and \eqref{def:P1}, respectively. Set
\begin{equation}\label{def:P2}
    \mathcal{P}_2(x):=\frac{1}{2\eta}\brac{\norm{x-\hat{x}}^2+\norm{\hat{x}-y}^2-2\sqrt{\frac{2\eta}{d}}\brac{\norm{x-\hat{x}}+\norm{\hat{x}-y}}-\frac{12\eta}{d}}. 
\end{equation}
Then, for every $x\in \R^d$, $   \Theta^{\eta,K}_y(x)\ge \mathcal{P}_1(x)\ge \mathcal{P}_2(x)$. In particular, the fact that $\Theta^{\eta,K}_y\geq \mathcal{P}_2$ ensures the acceptance test \eqref{eq:event_uniform:separation} is well-defined. 
\end{lemma}

Next, we bound the average number of proposals and the number of separation-oracle queries in Algorithm~\ref{alg:RGO:uniform:separation} per each iteration of Algorithm~\ref{alg:ASF_uniform}. The proof is deferred to Appendix~\ref{appendix:proofpropseparation}.

\begin{theorem}
    \label{theo:averagerejection_separation}
    Assume conditions (A1), (A2), and the stepsize $\eta =1/d^2$. Then regarding Algorithm~\ref{alg:RGO:uniform:separation}, 

    \begin{itemize}
        \item[a)] the average number of proposals in Algorithm \ref{alg:RGO:uniform:separation} is no more than $\sqrt{2\pi} M\exp\brac{\frac{13}{4}+\frac{20}{d}} + M \exp\brac{\frac{9}{4}+\frac{12}{d}}$. 
        \item[b)] there are at most ${\cal O}\left(d \log \frac{d \gamma}{\alpha}\right)$ queries to the separation oracle on $K$, where $\gamma=\frac{R}{\mbox{minwidth}(K)}$ and $\alpha\in (0,1)$ satisfies the concentration inequality $    \Pr\brac{\alpha\leq \frac{2}{7 d^3 R^2} }\leq 4 \exp\brac{-\frac{d^2R^2}{8}}$. 
    \end{itemize}
\end{theorem}

\section{Numerical experiments on dense $Z$-polytopes}
\label{section:numerical}

We compare our PS+Alg3 (Algorithm~\ref{alg:ASF_uniform}+Algorithm~\ref{alg:RGO:uniform:projection}) with In-and-Out (Algorithm~1 by  \cite{kook2024inandout}) on on random dense Z-polytopes. Per the explanation given in Appendix~\ref{appen_examplesconvexbodies}, a Z-polytope is specified by a generator matrix 
$V=[v_1,\dots,v_{m}]\in\mathbb{R}^{d\times m}$ through 
\[
Z \;=\;\{x=Vt:\ \|t\|_\infty\le 1\}
\;=\;\left\{\sum_{j=1}^{m} t_j v_j:\ t_j\in[-1,1]\right\}.
\]

\paragraph*{Setup.}
We draw $V\in\mathbb{R}^{d\times m}$ with i.i.d.\ standard Gaussian entries and
normalize each column to unit $\ell_2$ norm for numerical stability. We set
$d=40$ and $m=80$, and run $n=10^5$ outer iterations
per seed over $6$ independent seeds. Following Theorem~23 in
\cite{kook2024inandout}, we take $M=1$ and failure probability $p=0.1$,
which yields the conservative In-and-Out stepsize $\eta_{\mathrm{inout}}
=\brac{2d^{2} \log\!\Bigl(\frac{9  n  M}
{p}\Bigr)}^{-1}
\approx 1.95\times 10^{-5}.$

We will consider two stepsizes: 
\[
\eta_A := \eta_{\mathrm{inout}}\, \text{(Exp. A, conservative)} \quad\text{and}\quad \eta_B := 10^{-2}\, \text{(Exp. B, aggressive)}. 
\]

Note that the dominant scaling in $\eta_A$ is $1/d^2$; while the dependence on $n,M,p$ enters only through the logarithmic factor. In Experiment~A, we run both samplers with the conservative stepsize
$\eta_A$ that adheres to the recommendation of our Theorem~\ref{theo:averagerejection_projection} as well as  \cite[Theorem 23]{kook2024inandout} . In Experiment~B, we rerun both samplers at the larger practical $\eta_B$ to stress-test performance outside the conservative theory.
A baseline reference marginal is computed
once using a \emph{long hit-and-run chain}. All experiments are done in
\textsc{Matlab} with the Optimization Toolbox. The starting point of PS+Alg3 and In-and-Out is the central symmetry point of the Z-polytope. 

For Z-polytopes, $x\in\R^d$ is represented as $x=Vt$ with generator coefficients $t\in[-1,1]^{m}$. As explained in Appendix~\ref{appen_examplesconvexbodies}, both the membership check and projection for Z-polytopes are implemented via interior-point methods and have comparable arithmetic costs. Specifically, membership of $x$ is checked by LP feasibility (find $t\in[-1,1]^{m}$ with $Vt=x$, using \texttt{linprog}), and projection $\Pi_Z(y)$ is computed by the box-constrained QP $\min_{t\in[-1,1]^{m}}\|Vt-y\|_2^2$ (using \texttt{quadprog}) and returning $Vt^\star$. We report the elapsed time/wall-clock time, including all oracle-solver overhead. 

\begin{figure}[ht]
    \centering
    \begin{subfigure}[b]{0.48\linewidth}
        \centering
        \includegraphics[width=\linewidth]{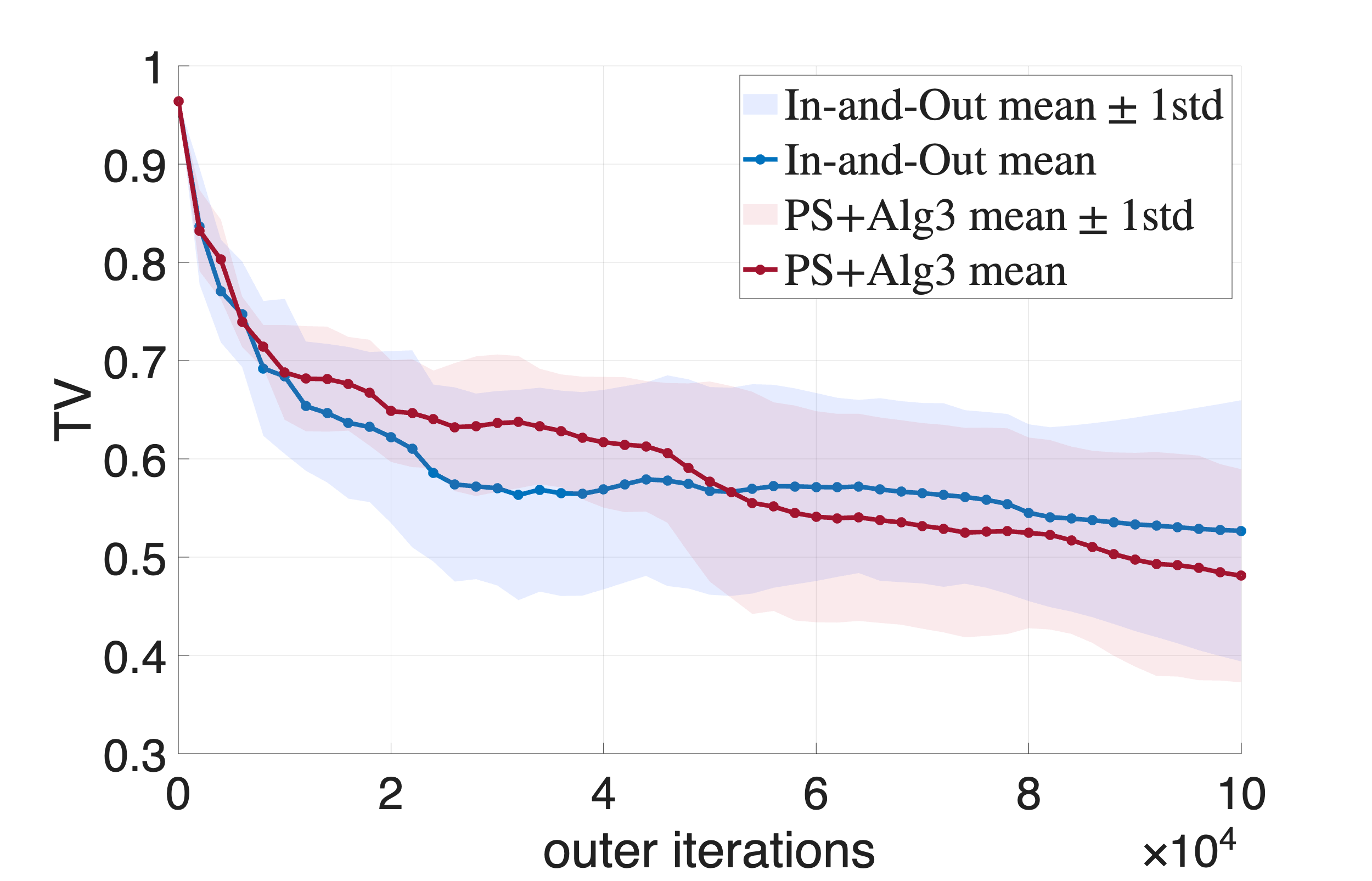}
        \caption{Exp. A ($\eta_A$): TV vs.\ iterations}
        \label{fig:zpoly-A-iter}
    \end{subfigure}
    \hfill
    \begin{subfigure}[b]{0.48\linewidth}
        \centering
        \includegraphics[width=\linewidth]{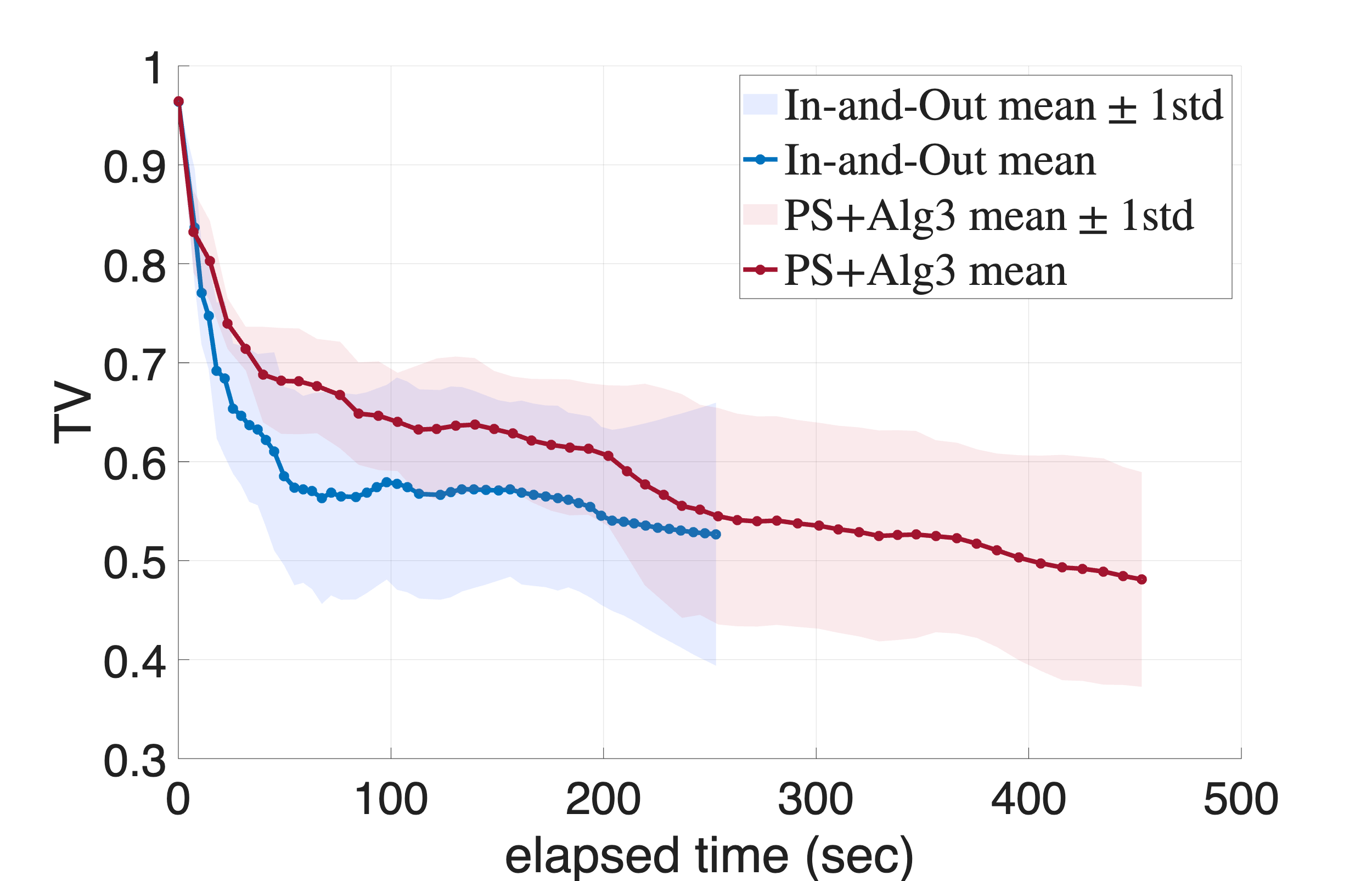}
        \caption{Exp. A ($\eta_A$): TV vs.\ elapsed time}
        \label{fig:zpoly-A-time}
    \end{subfigure}


    \begin{subfigure}[b]{0.48\linewidth}
        \centering
        \includegraphics[width=\linewidth]{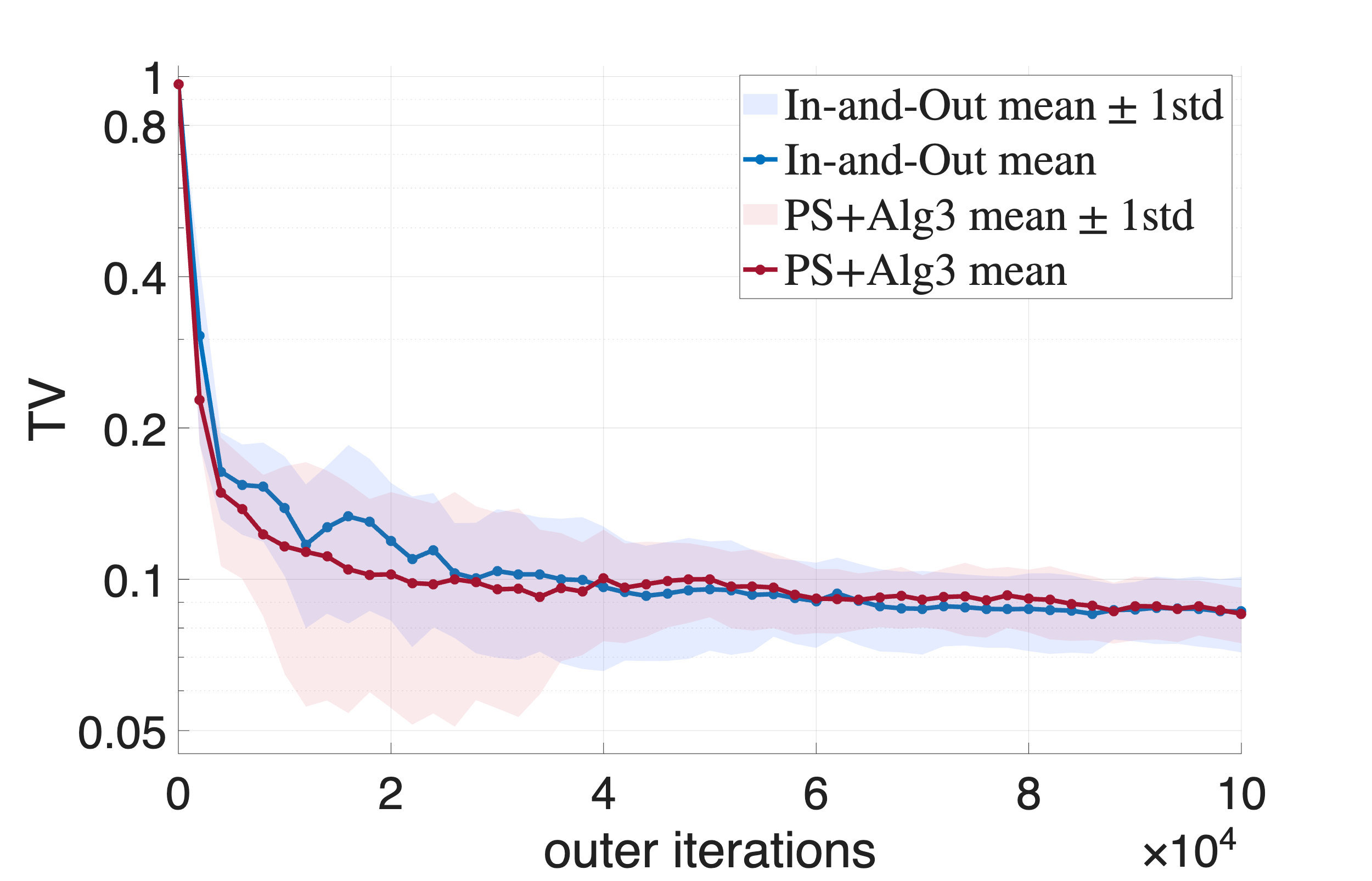}
        \caption{Exp. B ($\eta_B$): TV vs.\ iterations}
        \label{fig:zpoly-B-iter}
    \end{subfigure}
    \hfill
    \begin{subfigure}[b]{0.48\linewidth}
        \centering
        \includegraphics[width=\linewidth]{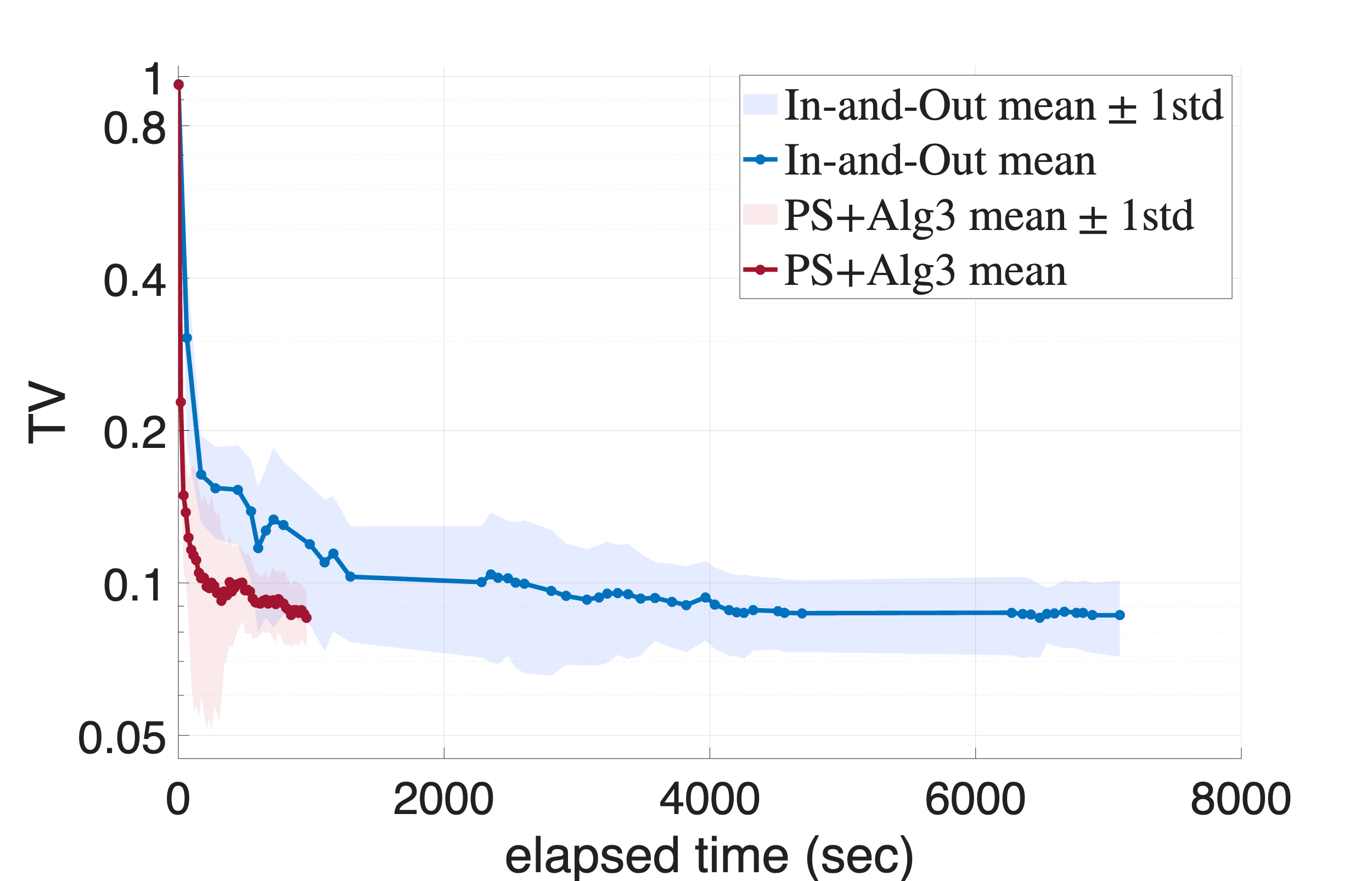}
        \caption{Exp. B ($\eta_B$): TV vs.\ elapsed time}
        \label{fig:zpoly-B-time}
    \end{subfigure}

    \caption{Comparison of In-and-Out and PS+Alg3.}
    \label{fig:zpoly-experiments-combined}
\end{figure}

\paragraph*{Results.}
Theorem~\ref{theo:outer} gives an outer-iteration bound in R\'enyi/$\chi^2$ divergence scaling as $1/\eta$. Instead, we plot the TV error of the $x_1$ marginal against a long hit-and-run reference chain, a numerically stable proxy for convergence, since estimating R\'enyi/$\chi^2$ divergence requires high-dimensional density-ratio estimates.

\emph{Experiment A (conservative stepsize $\eta_A$).}
Figure~\ref{fig:zpoly-A-iter} plots the TV error of the $x_1$ marginal versus outer iterations under the matched conservative stepsize $\eta_A$. From Figure~\ref{fig:zpoly-A-iter}, after about $5\times 10^4$ iterations, PS+Alg3 reduces TV faster and reaches a substantially lower TV level than In-and-Out. Meanwhile, Figure~\ref{fig:zpoly-A-time} shows TV versus elapsed time, where elapsed time captures both algorithmic complexity and oracle-solver overhead. The figure shows PS+Alg3 is slower in elapsed time, but it reaches a substantially lower TV level and thus has a better accuracy-time tradeoff.

\emph{Experiment B (aggressive stepsize $\eta_B$).}
Figure~\ref{fig:zpoly-B-iter} reports TV versus outer iterations of both
samplers at the larger stepsize $\eta_B=10^{-2}$.  In this instance, In-and-Out remains stable at $\eta_B$ and its TV curve
continues to decrease, reaching a final TV level comparable to PS+Alg3. PS+Alg3 still shows a modest
advantage in early iterations and attains a comparable final TV.
Figure~\ref{fig:zpoly-B-time} presents TV versus elapsed time in the same regime:
PS+Alg3 reaches a given TV substantially sooner in wall-clock time, whereas
In-and-Out requires longer elapsed time to achieve similar accuracy.


\vspace{-1.0em}
\section{Concluding remarks}
In this paper, we propose algorithms for uniform sampling from a convex body $K$ based on the proximal sampler. We explore the use of either the projection oracle on $K$ or the separation oracle on $K$, standard tools in convex optimization, for the RGO implementation (Algorithm \ref{alg:RGO:uniform:projection} and Algorithm~\ref{alg:RGO:uniform:separation}, respectively). Our RGO implementations are exact and therefore our algorithms do not have any failure probability. In both cases, the algorithms perform ${\cal O}(d^2)$ RGO steps. With a projection oracle, each RGO queries one projection and has at most ${\cal O}(1)$ expected proposals. With a separation oracle, each RGO queries ${\cal O}(d\log d)$ separations and has at most ${\cal O}(1)$ expected proposals. In addition, we perform numerical experiments to compare In-and-Out (Algorithm~1 by \cite{kook2024inandout}) with PS+Alg3 (Algorithm~\ref{alg:ASF_uniform}+Algorithm~\ref{alg:RGO:uniform:projection}) on random dense $Z$-polytopes. In the experiments of Section~\ref{section:numerical}, PS+Alg3 outperforms In-and-Out, both at the conservative stepsize suggested by the theory and at a more aggressive stepsize, even though the two methods invoke different oracles (projection versus membership).

We finally discuss some possible extensions of the paper. First, a natural question to ask beyond uniform sampling on $K$ is general log-concave sampling on $K$. Both uniform sampling on $K$ and log-concave sampling on $\R^d$ have benefited from using the proximal sampler as a generic framework in recent years; as a consequence, it is interesting to investigate algorithms based on the proximal sampler for sampling $\exp(-f(x))$ on $K$.
Second, for the purpose of uniform sampling on $K$, the RGO implementations in this paper (i.e., Algorithm \ref{alg:RGO:uniform:projection} and Algorithm \ref{alg:RGO:uniform:separation}) and 
those in \cite{kook2024inandout,kook2025coldstart,kook2025zeroth} all require a small stepsize $\eta=1/d^2$ so that RGO implementations within the proximal sampler remain efficient.
In contrast, for sampling from $\exp(-f(x))$ on $\R^d$, under the assumption that $f$ satisfies an $(L_\alpha, \alpha)$-semi-smooth condition for some $\alpha\in [0,1]$, the stepsize condition can be relaxed to $\eta= \tilde {\cal O}(d^{-\alpha/(\alpha+1)})$ in \cite{fan2023improved}. In particular, this improves the dimension dependence from ${\cal O}(d)$ to ${\cal O}(\sqrt{d})$ when $f$ is smooth (i.e., $\alpha=1$). However, techniques by \cite{fan2023improved} cannot be directly applied to uniform sampling on $K$, since the negative log-density (i.e., the indicator function $I_K(x)$) is discontinuous and hence lacks a smoothness notion. Therefore, reducing the dimension dependence for uniform sampling on $K$ still remains a challenging yet meaningful question.


\section*{Acknowledgement}
We thank Yunbum Kook for the helpful discussions.

\bibliographystyle{plain}
\bibliography{refs}

\newpage
\appendix
\begin{center}
\Large \bf Appendices to Oracle-based Uniform Sampling from Convex Bodies
\end{center}

The Appendices are organized as follows.

\begin{itemize}
\item Appendix \ref{appendix:uniformsampling} provides a summary of some algorithms for uniform sampling from convex bodies.

\item Appendix \ref{appendix:techlem} collects some technical lemmas and short proofs deferred from Sections~\ref{sec:prelim} and~\ref{sec:projection_uniform}. 
 
\item Appendix \ref{appendix:secseparation} contains results that are relevant to Algorithm \ref{alg:RGO:uniform:separation}, for example, the Cutting Plane method by \cite{jiang2020cuttingplane}.

\item Appendix~\ref{appen_examplesconvexbodies} presents examples of common convex bodies and their oracle implementations. 
\end{itemize}

\section{Algorithms for uniform sampling on convex bodies}
\label{appendix:uniformsampling}

Continuing the discussion in Section \ref{sec:intro}, we mention here several algorithms for uniform sampling from convex $K$. We refer to the survey \cite{vempala2005survey,vempala2010recent} and the dissertation \cite{cousins2017efficient} for additional details.

\begin{itemize}

\item Assuming a membership oracle, the Ball walk introduced by \cite{dyer1991random} works as follows: pick a uniform random point $y$ from the ball of radius $\delta$ centered at the
current point $x$; if $y$ is in $K$, go to $y$, otherwise stay at the current point $x$. Per \cite{lovasz1993random,vempala2005survey}, assuming the starting distribution is $M$-warm, the number of steps of the Ball walk to reach $\epsilon$-accuracy in the total variation distance is of the order $\mathcal{O}\brac{d^2 C_{\mathrm{LSI}}\frac{M^2}{\epsilon^2}\log\frac{M}{\epsilon} }$.

    \item The Hit-and-Run walk is first introduced in \cite{smith1984efficient} and rigorously investigated by Lov\'{a}sz and Simonovits in \cite{lovasz1993random}. Also assuming a membership oracle on $K$ and in the special case of uniform sampling, the Hit-and-Run walk is: choose a uniform direction over the unit sphere and find a line segment in that direction that intersects $K$ at two endpoints but still belong to $K$; then go to a uniform random point on that line segment. \cite{lovasz1999hit} shows the its iteration complexity in total variance is of the order $\mathcal{O}\brac{d^2C_{\mathrm{LSI}}\frac{M^2}{\epsilon^2}}$.
\item The coordinate Hit-and-Run walk \cite{turchin1971computation,diaconis2010gibbs,diaconis2012gibbs} is similar to the Hit-and-Run walk, with the difference being it picks a coordinate axis uniformly instead of considering all directions in a unit sphere. Although there have been experimental results in \cite{cousins2016practical,emiris2014efficient} which show the coordinate Hit-and-Run walk to mix faster than the original version in certain settings, the state-of-the-art upper bounds on the iteration complexity of the coordinate Hit-and-Run walk in \cite{laddha2023convergence,narayanan2022mixing} are worse than that of the original Hit-and-Run walk.

\item The In-and-Out algorithm by \cite{kook2024inandout} is the proximal sampler by \cite{leestructured21}, where the RGO implementation is via a membership oracle: sample $x$ from a Gaussian distribution until $x\in K$ up to a certain number of maximum attempts $N$, at which point the algorithm halts and declares failure. Their iteration complexities in R\'{e}nyi divergence and $\chi^2$-divergence are the same as those in our Theorem~\ref{theo:outer} in Section \ref{sec:prelim}, except that they derive their iteration complexities via the PI constant of the uniform distribution while we use both PI and LSI constants. 
\end{itemize}

\section{Technical results and proofs}\label{appendix:techlem}

First, we have some results about Gaussian integrals. 
\begin{lemma}\label{lem:gaussianint}
The following statements hold for any $\eta>0$, $c\in \R^d$ and $b\in \R$.
\begin{itemize}
    \item[(a)] $ \int_{\R^d} \exp\brac{-\frac{1}{2\eta} \norm{x-c}^2} \D x=(2\pi \eta)^{d/2}$;

    \item[(b)] $ \int_{\R^d} \exp\left(-\frac{1}{2\eta}(\|x-c\|-b)^2\right) \D x
    \le \exp \left(\frac14 + \frac{d}{\eta}b^2\right)(2\pi\eta)^{d/2}.$

\end{itemize} 
\end{lemma}

\begin{proof}
 The formula in Part a is a well-known fact about Gaussian integrals and thus the proof is omitted. Regarding Part b, let  $r:=\|x-c\|$. It holds for any $\theta\in(0,1]$ that $2cb\le \theta c^2+\frac{b^2}{\theta}$, which implies
\[
(r-b)^2 \ge (1-\theta)r^2 - \frac{1-\theta}{\theta}b^2.
\]
This combined with the formula in Part a lead to 
\[
\int_{\R^d} \exp\left(-\frac{1}{2\eta}(\|x-c\|-b)^2\right) \D x
\le (1-\theta)^{-d/2}\exp \left(\frac{1-\theta}{2\theta \eta}b^2\right)(2\pi\eta)^{d/2}. 
\]
With the choice $\theta = 1/(1+2d)$, we obtain
\[
\int_{\R^d} \exp\left(-\frac{1}{2\eta}(\|x-c\|-b)^2\right) \D x \le \left(1+\frac{1}{2d}\right)^{d/2}\exp \left(\frac{d}{\eta}b^2\right)(2\pi\eta)^{d/2},
\]
which gives the desired result by noting that $(1+1/2d)^{d/2} \le e^{1/4}$.
\end{proof}


The following result is applied in the proof of Lemma~\ref{lem:intR}. 
\begin{lemma} \label{lem:vol}
Denote $\partial K$ the boundary set of a non-empty, closed, and convex set $K \in \R^d$ such that $B(0,1) \subseteq K$. Then, we have
    \begin{equation} \label{ineq:vol}
		\vol_{d-1}(\partial K) \le d \vol(K).
	\end{equation}
\end{lemma}

\begin{proof}
    Consider any direction $v \in \R^d$, and denote the length from $0$ to $\partial K$ along $v$ as $L=L(v)>0$. Let $V(L)= \vol(K)$ and $A(L) = \vol_{d-1}(\partial K)$, then by the co-area formula, we have
	\[
	\vol(K) = V(L) = \int_0^L A(r) \D r.
	\]
	Differentiating both sides of the equation yields
	\[
	A(L) = V'(L)= V(1) d L^{d-1} = d \frac{V(L)}{L},
	\]
	where the last two identities use the fact that $V(L) = V(1) L^d$. 
	Now using $V(L)= \vol(K)$ and $A(L) = \vol_{d-1}(\partial K)$, we arrive at
	\[
		\vol_{d-1}(\partial K) = d \frac{\vol(K)}{L} \le d \vol(K),
	\]
	where the inequality follows from $B(0,1) \subseteq K$ and $L=L(v)\ge 1$ for any direction $v$.
\end{proof}







In Algorithm~\ref{alg:RGO:uniform:projection} and Algorithm~\ref{alg:RGO:uniform:separation}, we use a simple rejection-sampling scheme: propose from a tractable density proportional to $\exp(-\mathcal{P})$, and accept with the standard likelihood-ratio rule in order to target a density proportional to $\exp(-\Theta)$.
In the upcomoing lemma, we record the resulting target law, the acceptance probability, and the average number of proposals.

\begin{lemma}\label{lem:rejection} 
Let $\Theta,\mathcal{P}:\R^{m}\to\R\cup\{+\infty\}$ be measurable and assume
\[
0<N_\Theta:=\int_{\R^m} e^{-\Theta(z)}  \D z<\infty,
\qquad
0<N_{\mathcal{P}}:=\int_{\R^m} e^{-\mathcal{P}(z)}  \D z<\infty,
\]
and that $\Theta(z)\ge \mathcal{P}(z)$ for all $z\in\R^m$.
Consider the rejection sampler that repeats the following trial until acceptance:
\begin{enumerate}
    \item draw a proposal $Z\in\R^m$ with density $k(z)=e^{-\mathcal{P}(z)}/N_{\mathcal{P}}$;
    \item draw $U\sim{\cal U}[0,1]$ independent of $Z$ and accept if and only if
    \begin{equation}\label{eq:generic-rejection-event}
        U\le \exp\bigl(-\Theta(Z)+\mathcal{P}(Z)\bigr).
    \end{equation}
\end{enumerate}
Then the output has density
\[
\pi(z)=\frac{e^{-\Theta(z)}}{N_\Theta}.
\]
Moreover, let $E$ denote the acceptance event \eqref{eq:generic-rejection-event} for a single trial and $F$ denote the number of trials until acceptance. Then, the acceptance probability and the average number of proposals are
\begin{equation}\label{eq:generic-accept-rate}
p:=\Pr(E)=\frac{N_\Theta}{N_{\mathcal{P}}},
\qquad
\mathbb{E}[F] =\frac{1}{p}=\frac{N_{\mathcal{P}}}{N_\Theta}.
\end{equation}
\end{lemma}

\begin{proof}
For one trial, let $E$ denote the event \eqref{eq:generic-rejection-event} happens.
Since $\Theta\ge \mathcal{P}$, we have $0\le \Pr(E\mid Z=z)=\exp(-\Theta(z)+\mathcal{P}(z))\le 1$,
so the acceptance rule is well-defined.

We first compute the acceptance probability:
\begin{align*}
p=\Pr(E)
=\int_{\R^m}\Pr(E\mid Z=z) k(z)  \D z &=\int_{\R^m}\exp\bigl(-\Theta(z)+\mathcal{P}(z)\bigr)\cdot \frac{e^{-\mathcal{P}(z)}}{N_{\mathcal{P}}}  \D z \\
&=\frac{1}{N_{\mathcal{P}}}\int_{\R^m}e^{-\Theta(z)}  \D z
=\frac{N_\Theta}{N_{\mathcal{P}}}.
\end{align*}
Next, we compute the conditional density of the accepted proposal:
\begin{align*}
k(z\mid E)
&=\frac{\Pr(E\mid Z=z) k(z)}{\Pr(E)} =\frac{\exp\bigl(-\Theta(z)+\mathcal{P}(z)\bigr)\cdot e^{-\mathcal{P}(z)}/N_{\mathcal{P}}}{N_\Theta/N_{\mathcal{P}}}
=\frac{e^{-\Theta(z)}}{N_\Theta}.
\end{align*}
Thus the accepted proposal has density $\pi(z)\propto e^{-\Theta(z)}$, as claimed.

Finally, each trial is i.i.d. and succeeds with probability $p$, hence $F\sim\mathrm{Geom}(p)$ and
$\E[F]=1/p$.
\end{proof}

Next we record a straightforward implication of the warm start assumption (A2). 
\begin{lemma}
\label{lem:warmstart}
Let $\mu^X_k$ denote the law of $x_k$ and $\mu^k$ denote the law of $y_k$ in Algorithm~\ref{alg:ASF_uniform}.
Under the warm start condition (A2), the warmness is preserved for both the $x$ and $y$-updates at every iteration $k\geq 0$:
\[\frac{\D \mu^k}{\D \pi^Y}\leq M \quad\text{and}\quad \frac{\D \mu^X_k}{\D \pi^X} \leq M. \]
\end{lemma}
\begin{proof}
We proceed by induction. The base case $k=0$ holds by assumption (A2).
Assume as the induction hypothesis that $\frac{\D \mu^X_k}{\D \pi^X} \leq M$.

\textbf{Step 1:}
Recall $\mu^k$ is the distribution of $y_k$ obtained in Step 1 of Algorithm~\ref{alg:ASF_uniform}.
For any measurable set $U\subseteq \R^d$, let $U-y:=\{x\in \R^d:x+y\in U \}$. Denote $\gamma(\cdot)$ the density of $\mathcal{N}(0,\eta I_d)$.
Since $\mu^k = \mu^X_k * \gamma$ and $\pi^Y = \pi^X * \gamma$ where $*$ denotes the convolution of the probability measures, we have
\begin{align*}
    \mu^k(U)= \int_{\R^d} \mu^X_k(U-y) \gamma(y)\D y&=\int_{\R^d} \left(\int_{U-y} \frac{\D \mu^X_k}{\D \pi^X}(t) \pi^X(\D  t) \right) \gamma(y) \D y\\
    &\leq M\int_{\R^d} \pi^X(U-y)\gamma(y) \D y = M\pi^Y(U),
\end{align*}
where we have applied the induction hypothesis to get the last line. This implies $\frac{\D \mu^k}{\D \pi^Y} \leq M$. 

\textbf{Step 2:}
Let $\mu^X_{k+1}$ denote the distribution of $x_{k+1}$ obtained in Step 2 of Algorithm~\ref{alg:ASF_uniform}.
Recall that $x_{k+1}$ is sampled from the conditional distribution $\pi^{X|Y}(\cdot|y_k)$.
For any measurable set $A \subseteq K$, we have
\begin{align*}
\mu^X_{k+1}(A) &= \int_{\R^d} \pi^{X|Y}(A|y)  \mu^k(\D y) \\
&= \int_{\R^d} \pi^{X|Y}(A|y) \frac{\D \mu^k}{\D \pi^Y}(y)  \pi^Y(\D y ) \\
&\leq \int_{\R^d} \pi^{X|Y}(A|y) M  \pi^Y(\D y) = M \int_{\R^d} \pi^{X,Y}(A , \D y) = M \pi^X(A), 
\end{align*}
where we apply $\frac{\D \mu^k}{\D \pi^Y} \leq M$ from Step~1 to get the second to last line. Thus, $\frac{\D \mu^X_{k+1}}{\D \pi^X} \leq M$. The induction proof is complete. 
\end{proof}


Recall that Lemma~\ref{lem:intR} shows the following relation
\[
\int_{K^c} \exp\left(-\frac{\brac{\|\proj_K(y)-y\|-\tau}^2}{2\eta}\right) \D y=\int_0^\infty \exp\left(-\frac{\brac{\delta-\tau}^2}{2\eta}\right) \vol_{d-1}(\partial K_\delta) \D \delta,
\]
and provides an upper bound on it.
The following result presents an alternative upper bound.

\begin{lemma}
\label{lem:alternativeintbypart}
Assume the setup in Lemma~\ref{lem:intR} and write $v(\delta)= \mathrm{vol}(K_\delta)$ and $a(\delta) = \mathrm{vol}_{d-1}(\partial K_\delta)$. Then, $v'(\delta)=a(\delta)$ for almost every
$\delta\ge 0$.
Fix $\eta>0$ and $\tau\ge 0$, and set
\[
  w_{\eta,\tau}(\delta)
  = \exp\!\left(-\frac{(\delta-\tau)^2}{2\eta}\right).
\]
It holds that
\[
\int_0^\infty a(\delta) w_{\eta,\tau}(\delta) \D\delta
\le
\vol(K)\Bigl(e^{d\tau}-1\Bigr)
+
\vol(K)\sqrt{2\pi\eta d^2}\;
\exp\!\left(\tau d+\frac{\eta d^2}{2}\right).
\]
\end{lemma}

\begin{proof}
By the co-area formula, we have $v(\delta) = \mathrm{vol}(K_\delta) = \int_0^{\delta} \mathrm{vol}_{d-1}(\partial K_s) \D s = \int_0^{\delta} a(s) \D s$, so that  $v'(\delta)=a(\delta)$ for almost every $\delta\ge 0$. Then integrating by parts gives
\begin{align}
\label{intbypart}
 \int_0^\infty a(\delta) w_{\eta,\tau}(\delta)  \D \delta=  \int_0^\infty v'(\delta) w_{\eta,\tau}(\delta)  \D \delta
  =
  v(\delta)w_{\eta,\tau}(\delta)\bigg|_{0}^{\infty}
  - \int_0^\infty v(\delta) w_{\eta,\tau}'(\delta)  \D \delta.
\end{align}
Noting that $v(\delta)w_{\eta,\tau}(\delta)\to 0$ as $\delta\to\infty$, we have
\begin{equation}
\label{step_intbypart}
  \int_0^\infty a(\delta) w_{\eta,\tau}(\delta)  \D \delta
  =
  -\mathrm{vol} (K)\exp\!\left(-\frac{\tau^2}{2\eta}\right)
  + \int_0^\infty v(\delta) \frac{\delta-\tau}{\eta} w_{\eta,\tau}(\delta)  \D \delta.
\end{equation}
It follows from the observation $K\subseteq K_\delta\subseteq (1+\delta)K$ that
\begin{equation}
\label{eq:vdeltacomparedtov0}
  \vol(K)\le v(\delta)\le (1+\delta)^d \mathrm{vol} (K)\le e^{ d \delta} \mathrm{vol} (K).
\end{equation}
For $\delta \in [0,\tau]$, we have
\[
\int_0^\tau v(\delta)\frac{\delta-\tau}{\eta}w_{\eta,\tau}(\delta) \D\delta
\stackrel{\eqref{eq:vdeltacomparedtov0}}\le
\vol(K)\int_0^\tau \frac{\delta-\tau}{\eta}w_{\eta,\tau}(\delta) \D\delta
=
\vol(K)\Bigl(-1+e^{-\tau^2/(2\eta)}\Bigr),
\]
where the last identity uses $w'_{\eta,\tau}(\delta)=-\frac{\delta-\tau}{\eta}w_{\eta,\tau}(\delta)$.
For $\delta \in [\tau,\infty)$, using \eqref{eq:vdeltacomparedtov0}, we obtain
\[
\int_\tau^\infty v(\delta)\frac{\delta-\tau}{\eta}w_{\eta,\tau}(\delta) \D\delta
\le
\vol(K)\int_\tau^\infty e^{d\delta}\frac{\delta-\tau}{\eta}w_{\eta,\tau}(\delta) \D\delta.
\]
Plugging the above two inequalities into \eqref{step_intbypart} yields
\begin{equation}
\int_0^\infty a(\delta) w_{\eta,\tau}(\delta) \D\delta
=
-\vol(K)
+\vol(K)\int_\tau^\infty e^{d\delta}\frac{\delta-\tau}{\eta}w_{\eta,\tau}(\delta) \D\delta.
\label{step_intbypart_split}
\end{equation}
Let $F(\delta):=e^{d\delta}w_{\eta,\tau}(\delta)$, then we note that
\[
  \frac{\delta-\tau}{\eta}F(\delta)=dF(\delta)-F'(\delta).
\]
It follows from $F(\delta)\to 0$ as $\delta\to\infty$ and $F(\tau)=e^{d\tau}$ that
\begin{align}
\int_\tau^\infty e^{d\delta}\frac{\delta-\tau}{\eta}w_{\eta,\tau}(\delta) \D\delta
&=
\int_\tau^\infty \Bigl(dF(\delta)-F'(\delta)\Bigr) \D\delta \nonumber\\
&=
d\int_\tau^\infty e^{d\delta}w_{\eta,\tau}(\delta) \D\delta
-F(\delta)\bigg|_{\tau}^{\infty}
=
d\int_\tau^\infty e^{d\delta}w_{\eta,\tau}(\delta) \D\delta
+e^{d\tau}.
\label{eq:kindofibp_alt}
\end{align}
Plugging \eqref{eq:kindofibp_alt} into \eqref{step_intbypart_split} yields
\begin{align}
\int_0^\infty a(\delta) w_{\eta,\tau}(\delta) \D\delta
\le
\vol(K)\Bigl(e^{d\tau}-1\Bigr)
+
d \vol(K)\int_\tau^\infty e^{d\delta}w_{\eta,\tau}(\delta) \D\delta.
\label{step_reduce_to_gaussian_tail}
\end{align}
Finally, we complete the square by writing $d\delta-\frac{(\delta-\tau)^2}{2\eta}
=
-\frac{\left(\delta-(\tau+\eta d)\right)^2}{2\eta}
+\tau d+\frac{\eta d^2}{2}$,
so
\begin{align}
\int_\tau^\infty e^{d\delta}w_{\eta,\tau}(\delta) \D\delta
&\le
\int_{-\infty}^{\infty}\exp\!\left(d\delta-\frac{(\delta-\tau)^2}{2\eta}\right) \D\delta \nonumber\\
&=
\exp\!\left(\tau d+\frac{\eta d^2}{2}\right)
\int_{-\infty}^{\infty}\exp\!\left(-\frac{\left(\delta-(\tau+\eta d)\right)^2}{2\eta}\right) \D\delta \nonumber\\
&=
\sqrt{2\pi\eta}\;
\exp\!\left(\tau d+\frac{\eta d^2}{2}\right),
\end{align}
where the last identity follows from Lemma~\ref{lem:gaussianint}(a).
Combining this with \eqref{step_reduce_to_gaussian_tail} gives, for every $\tau\ge 0$, the desired bound on $\int_0^\infty a(\delta) w_{\eta,\tau}(\delta) \D\delta$. This completes the proof.
\end{proof}

We provide two missing proofs in Section~\ref{sec:prelim} and Section~\ref{sec:projection_uniform} below. 

\textbf{Proof of Lemma~\ref{lem:compareP1}}
It follows from the convexity of $K$ (condition~(A1)) that 
\[\inner{x-\proj_K(y)}{\proj_K(y)-y}\geq 0,\]
and hence that the first stated inequality in Lemma~\ref{lem:compareP1} holds. 
In view of \eqref{regularizedmap} and \eqref{def:P1}, we observe that the RHS of \eqref{eq:event_uniform} is equivalent to
\[
\exp\left(-I_K(x)-\frac{1}{\eta}\inner{x-\proj_K(y)}{\proj_K(y)-y}\right) = \exp(\mathcal{P}_1(x)- \Theta_y^{\eta,K}(x)).
\]
Hence, the proof is completed.
\QEDA


At this point, we recall an important result by~\cite{kook2024inandout} about contractivity in R\'enyi and $\chi^2$ divergences for Algorithm~\ref{alg:ASF_uniform}. Using a clever smoothing argument, \cite{kook2024inandout} adapts the approach of \cite{chen2022improved} to uniform sampling over a convex body $K$. Moreover, as noted in \cite{kook2024inandout}, the result does not actually require $K$ to be convex.

\begin{theorem}
    \label{theo:kooketal}
    (\cite[Theorem~23]{kook2024inandout}) Let $\mu^X_k$ be the law of the $k$-th output of Algorithm \ref{alg:ASF_uniform}. Denote $C_{\mathrm{PI}}$ and $C_{\mathrm{LSI}}$ respectively the Poincar\'{e} constant and the log Sobolev constant of the uniform distribution $\pi^X$ on $K$  whose asymptotics are provided in Lemma~\ref{lem:isoconstant}. Then for any $q\geq 1$,
\begin{equation}\label{rate}
    \mathcal{R}_q\brac{\mu^X_k||\pi^X}\leq \frac{\mathcal{R}_q\brac{\mu^X_0||\pi^X}}{\brac{1+{\eta}/{C_{\mathrm{LSI}}}}^{2k/q}},\qquad  \chi^2\brac{\mu^X_k||\pi^X}\leq \frac{\chi^2\brac{\mu^X_0||\pi^X}}{\brac{1+{\eta}/{C_{\mathrm{PI}}}}^{2k}}. 
\end{equation}
\end{theorem}

\vspace{3mm}

Now, we are ready to prove Theorem~\ref{theo:outer} based on the contraction result in Theorem~\ref{theo:kooketal}. 

\vspace{3mm}

\noindent
\textbf{Proof of Theorem~\ref{theo:outer}:}
The $M$-warm start assumption implies $\mathcal{R}_q\brac{\mu^X_0||\pi^X}\leq \frac{q}{q-1}\log M$. Then via the first part of \eqref{rate} in Theorem~\ref{theo:kooketal},  we can solve for 
    \[
        \frac{\frac{q}{q-1}\log M}{\brac{1+\eta/C_{\mathrm{LSI}}}^{2k/q}} \stackrel{\eqref{rate}}\leq \epsilon
    \]
to get 
\[
   k\geq \frac{q}{2}\frac{\log \brac{\frac{q}{q-1}\frac{\log M}{\epsilon}} }{\log\brac{1+1/C_{\mathrm{LSI}}}\eta}\geq  \frac{q}{2}\log \brac{\frac{q}{q-1}\frac{\log M}{\epsilon}} \frac{C_{\mathrm{LSI}}}{\eta}.
\]
Therefore, we can take $k=O\brac{\frac{1}{\eta}{C_{\mathrm{LSI}}q\log \brac{2\frac{\log M}{\epsilon}} }}$ as in Theorem~\ref{theo:outer}(a). The calculation for $\chi^2$ divergence is along the same line with the use of the second part of \eqref{rate} in Theorem~\ref{theo:kooketal}. This completes the proof. 
\QEDA

\section{Supplementary materials for Section \ref{sec:separation_uniform}}
\label{appendix:secseparation}

\subsection{Results about the Cutting Plane method by \cite{jiang2020cuttingplane}}
\label{appendix:cuttingplane}

We first restate \cite[Theorem~C.1]{jiang2020cuttingplane}, which is about the iteration complexity and running time of the Cutting Plane method by \cite{jiang2020cuttingplane}. 

\begin{theorem}(\cite[Theorem~C.1]{jiang2020cuttingplane})\label{theo:jiang} Let $f$ be a convex function on $\R^d$. $K$ is a convex set that contains a minimizer
of $f$ and $K\subseteq B_\infty (0,R)$, where $B_\infty (0,R)$ denotes a ball of radius $R$ in $\ell_\infty$ norm, i.e., $\norm{x}_\infty=\sup_{1\leq i\leq d}\abs{x_i}$. 

Suppose we have a subgradient oracle for $f$ with cost $T$ and a separation oracle for $K$ with cost $S$. Using $B_\infty (0,R)$ as the initial polytope for our Cutting Plane Method, for any $0< \alpha<1$, we can compute $\hat x\in K$ such that 
\begin{equation}\label{ineq:opt}
f(\hat x)-\min_{x \in K} f(x) \le \alpha \left(\max_{x \in K} f(x)-\min_{x \in K} f(x)\right).    
\end{equation}
with a running time of $ {\cal O}\left(T\cdot d \log \frac{d \gamma}{\alpha}+S\cdot  d \log \frac{d \gamma}{\alpha}+d^3\log \frac{d \gamma}{\alpha}\right)$. In particular, the number of subgradient oracle calls and the number of separation oracle calls are of the order 
\begin{align*}
    \mathcal{O}\brac{d \log \frac{d \gamma}{\alpha}}, 
\end{align*}
where $\gamma=\frac{R}{\mbox{minwidth}(K)}$.
\end{theorem}

Next, we apply the above theorem to find a $(1/d)$-solution to $ \argmin_{x\in K} \frac{1}{2\eta}\norm{x-y}^2$, a sub-problem that appears in Algorithm \ref{alg:RGO:uniform:separation}. Since the iteration complexity involves the constant $\alpha$ to be chosen below, we also provide a concentration inequality to show $\alpha$ does not adversely affect the iteration complexity in high probability.

\begin{lemma}
\label{lem:cuttingplanealg:deltasolution} Assume condition (A1) holds. For given $y\in \R^d$, set
\begin{equation}\label{eq:alpha}
    \alpha= \frac{2}{d^3\brac{R^2 + 2R\|\proj_K(y)-y\|}}
\end{equation}    
and $\gamma=\frac{R}{\mbox{minwidth}(K)}$.
    Moreover, assume there is a separation oracle for $K$. Then the Cutting Plane method by \cite{jiang2020cuttingplane} makes ${\cal O}\left(d \log \frac{d \gamma}{\alpha}\right)$ separation oracle calls to generate a $(1/d)$-solution $\hat{x}\in K$ to the optimization problem $\min_{x\in K} \frac{1}{2\eta}\norm{x-y}^2$. 
    In addition, we provide the following bound on $\alpha$ with high probability
   \begin{equation}\label{ineq:prob}
    \Pr\brac{\alpha\leq \frac{2}{d^37R^2} }\leq 4 \exp\brac{\frac{-R^2}{8\eta}}. 
\end{equation}
\end{lemma}

\begin{proof}
Since $K$ is closed per condition (A1), $K$ contains a minimizer of $f$. Moreover, the fact that $K$ is contained the Euclidean ball $B(0,R)$ per condition (A1) implies $K$ is also contain in the ball $B_\infty(0,R)$. Then to be able to apply Theorem~\ref{theo:jiang}, we need to verify that 
\begin{align}
\label{claim:alphalessthan1}
    0\leq \alpha= \frac{2}{d^3\brac{R^2 + 2R\|\proj_K(y)-y\|}}<1
\end{align}
and that
\begin{align}
\label{claim:deltasolution}
   \alpha \left(\max_{x \in K} f(x)-\min_{x \in K} f(x)\right) \le 1/d. 
\end{align}
Then Theorem~\ref{theo:jiang} guarantees that the Cutting Plane method by \cite{jiang2020cuttingplane} produces a $(1/d)$-solution with $\mathcal{O}\brac{d \log \frac{d \gamma}{\alpha}}$ separation oracle calls.

Since $B(0,1) \subseteq K \subseteq B(0,R)$, we have $\alpha = \frac{2 }{d^3\brac{R^2 + 2R\|\proj_K(y)-y\|}} \le \frac{2}{d^3 R^2} \le \frac{2}{d^3} < 1$,
and thus \eqref{claim:alphalessthan1} is true. 

Next, let us set $x^*=\max_{x\in K} \frac{1}{2\eta}\norm{x-y}^2$. Then, using the triangle inequality, we have
\begin{align*}
     &\alpha \left(\max_{x \in K} f(x)-\min_{x \in K} f(x)\right)= \frac{\alpha}{2\eta} \left(\|x^* - y\|^2 - \|\proj_K(y) - y\|^2\right)\\
    \le& \frac{\alpha}{2\eta} \left[(\|x^* - \proj_K(y)\| + \|\proj_K(y) - y\|)^2 - \|\proj_K(y) - y\|^2\right]\\
    \le& \frac{\alpha}{2\eta} (R^2 + 2R\|\proj_K(y) - y\|) \stackrel{\eqref{eq:alpha}}= \frac{1}{d},
\end{align*}
where the last identity follows from $\eta=1/d^2$ and the definition of $\alpha$ in \eqref{eq:alpha}. Hence, \eqref{claim:deltasolution} is true. 

Regarding the concentration inequality, recall that $y$ is the output of step 1 in the proximal sampler (Algorithm~\ref{alg:ASF_uniform}) and satisfies $ y=y_k=x_{k-1} +\sqrt{\eta}Z$, where $Z\sim {\cal N}(0,I)$ and $k$ denotes some iterate of the proximal sampler. Then, we can write
\begin{align*}
   \norm{\proj_K(y)-y}&=\norm{\proj_K(y)-x_{k-1}-\sqrt{\eta}Z}\\
   &\leq \norm{\proj_K(y)}+\norm{x_{k-1}}+\sqrt{\eta}\norm{Z}\leq 2R+\sqrt{\eta}\norm{Z}. 
\end{align*}
The last inequality is due to $B(0,R) \supset K$ and $\proj_K(y),x_{k-1}\in K$. Combining with the Gaussian concentration inequality from \cite[Equation (3.5)]{ledoux2013probability} to get
\begin{align*}
    \Pr\brac{ \norm{\proj_K(y)-y}>3R}\leq \Pr\brac{\norm{Z}\geq R/\sqrt{\eta}}\leq 4\exp\brac{-R^2/(8\eta)}. 
\end{align*}
This together with $\alpha$ in \eqref{eq:alpha} implies that \eqref{ineq:prob} holds and completes the proof. 
\end{proof}

\subsection{Proof of Lemma~\ref{lem:comparealltheP}}
\label{proof:alltheP}

    In view of Lemma~\ref{lem:compareP1}, the fact that $\Theta^{\eta,K}_y(x) \ge \mathcal{P}_1(x)$ for every $x\in \R^d$ immediately holds under condition~(A1), so what remains is to show $\mathcal{P}_1(x)\geq \mathcal{P}_2(x)$.

\noindent
Recall that $\hat{x}$, the $(1/d)$-solution to $\min_{x\in K}\{f(x):=\norm{x-y}^2/(2\eta)\}$, obtained by the Cutting Plane method by \cite{jiang2020cuttingplane} belongs to $K$. Since $f$ is $\eta^{-1}$-strongly convex and $K$ is a convex set, we have 
    \begin{align}
    \label{bound:barxkandx*}
        \norm{\hat{x}-\proj_K(y)}\leq \sqrt{2\eta \brac{f(\hat{x})-f(\proj_K(y))}}\leq \sqrt{\frac{2\eta}{d}}.  
    \end{align}
This inequality and the triangle inequality imply that for any generic $x\in \R^d$,
\begin{equation}
\label{consequence:triangle}
    \norm{x-\hat{x}}\leq \norm{x-\proj_K(y)}+\norm{\proj_K(y)-\hat{x}}\leq \norm{x-\proj_K(y)}+\sqrt{\frac{2\eta}{d}}. 
\end{equation}
It follows that 
\begin{align*}
    \norm{x-\hat{x}}^2
    &\leq \norm{x-\proj_K(y)}^2+2\norm{x-\proj_K(y)}\sqrt{\frac{2\eta}{d}}+\frac{2\eta}{d}\\
    &\leq \norm{x-\proj_K(y)}^2+2\brac{\norm{x-\hat{x}}+\norm{\hat{x}-\proj_K(y)} }\sqrt{\frac{2\eta}{d}}+\frac{2\eta}{d}\\
    &\stackrel{\eqref{bound:barxkandx*}}\leq \norm{x-\proj_K(y)}^2+2\norm{x-\hat{x}}\sqrt{\frac{2\eta}{d}}+\frac{6\eta}{d},
\end{align*}
The above inequality can be rearranged as 
\begin{align}
\label{secondone}
    \norm{x-\hat{x}}^2-2\norm{x-\hat{x}}\sqrt{\frac{2\eta}{d}}-\frac{6\eta}{d}\leq \norm{x-\proj_K(y)}^2.
\end{align}
Similarly, 
\begin{align}
    \label{thirdone}
    \norm{y-\hat{x}}^2-2\norm{y-\hat{x}}\sqrt{\frac{2\eta}{d}}-\frac{6\eta}{d}\leq \norm{y-\proj_K(y)}^2.
\end{align}
Combining \eqref{secondone} and \eqref{thirdone} leads to the desired conclusion that $\mathcal{P}_1(x)\geq \mathcal{P}_2(x)$. 
\QEDA

\subsection{Proof of Theorem~\ref{theo:averagerejection_separation}}
\label{appendix:proofpropseparation}
Denote $\mu^k$ the distribution of $y=y_k$ at the k-th  iteration of Algorithm \ref{alg:ASF_uniform}. Write $n_y$ the average number of proposals conditioned on $y$ in the rejection sampler employed in Algorithm~\ref{alg:RGO:uniform:projection}. Then the average number of proposals is $\E_{\mu^k}[n_y]$ where the formula of $n_y$ is given in \eqref{eq:generic-accept-rate}. The fact that $d\mu^k/d\pi^Y\leq M$ per Lemma~\ref{lem:warmstart} implies
\begin{equation}
\label{warmstartimplication_separation}
	\E_{\mu^k}[n_y] \le M \E_{\pi^Y}[n_y], 
\end{equation}
and hence we will focus on bounding $\E_{\pi^Y}[n_y]$. 
In view of \eqref{eq:generic-accept-rate}, the latter expression becomes
    \[
	\E_{\pi^Y}[n_y] = \int_{\R^d} \frac{\int_{\R^d} \exp\brac{-{\mathcal{P}_2}(x)} \D x}{\int_K \exp\left(-\frac{1}{2\eta}\|x-y\|^2\right) \D x} \pi^Y(y) \D y. 
    \]
 Using the formula for $\pi^Y$ in \eqref{for:piy}, we get
\begin{equation*} 
    \E_{\pi^Y}[n_y] \le \frac{1}{\vol(K)(2\pi\eta)^{d/2}}\int_{\R^d} \int_{\R^d} \exp\brac{-{\mathcal{P}_2}(x)} \D x \D y.
\end{equation*}
Let us define an auxiliary function
\begin{equation}\label{def:P3}
    \mathcal{P}_3(x):= 
    \frac{1}{2\eta}\brac{\brac{\norm{x-\hat{x}}-\sqrt{\frac{2\eta}{d}} }^2+\brac{\norm{y-\proj_K(y)}-2\sqrt{\frac{2\eta}{d}}}^2-\frac{32\eta}{d} }. 
\end{equation}
We can easily show at the end of this proof that $\mathcal{P}_2(x)\ge \mathcal{P}_3(x)$, which leads to 
\begin{equation} \label{averagerejection:separation:middleestimate}
    \E_{\pi^Y}[n_y] \le \frac{1}{\vol(K)(2\pi\eta)^{d/2}}\int_{\R^d} \int_{\R^d} \exp\brac{-{\mathcal{P}_3}(x)} \D x \D y.
\end{equation}
The definition of ${\mathcal{P}_3}$ in \eqref{def:P3} and Part b of Lemma~\ref{lem:gaussianint} imply 
\begin{align}
\label{int:expPseparation}
    &\int_{\R^d} \exp\brac{-\mathcal{P}_3(x)} \D x \nonumber\\
    &\stackrel{\eqref{def:P3}}{=} \exp\left(-\frac{(\|\proj_K(y)-y\|-2\sqrt{2\eta/d})^2}{2\eta} + \frac{16}{d}\right) \int_{\R^d} \exp\left(-\frac{(\|x-\hat x\| - \sqrt{2\eta/d})^2 }{2\eta}\right) \D x \nonumber\\
    &\stackrel{\text{Lemma~\ref{lem:gaussianint},b}}{\le} \exp\left(-\frac{1}{2\eta}\brac{\|\proj_K(y)-y\|-2\sqrt{\frac{2\eta}{d}}}^2 + \frac{16}{d}+\frac{9}{4}\right)(2\pi\eta)^{\frac d2}. 
\end{align}
Next, we combine the previous calculations and Lemma~\ref{lem:intR} with $\tau=2\sqrt{2\eta/d}$ to get
\begin{align*}
     \E_{\pi^Y}[n_y]&\stackrel{\eqref{averagerejection:separation:middleestimate},      \eqref{int:expPseparation}}{\leq} \frac{1}{\vol(K)} \exp\brac{\frac{16}{d}+\frac{9}{4}}\int_{\R^d} \exp\left(-\frac{1}{2\eta}\brac{\|\proj_K(y)-y\|-2\sqrt{\frac{2\eta}{d}}}^2\right)\D y\\
     &\stackrel{\text{Lemma~\ref{lem:intR}}}{\le} \exp\brac{\frac{9}{4}+\frac{16}{d}} \exp\left(\frac{\eta d^2}{2}+2\sqrt{2\eta d}\right)\sqrt{2\pi \eta d^2}+\exp\brac{\frac{9}{4}+\frac{12}{d}}. 
\end{align*}
Plugging $\eta=1/d^2$ into the above formula yields
\begin{align*}
     \E_{\pi^Y}[n_y] &\le \exp\brac{\frac{9}{4}+\frac{16}{d}}\exp\brac{\frac{1}{2}+2\sqrt{\frac{2}{d}}}\sqrt{2\pi}+\exp\brac{\frac{9}{4}+\frac{12}{d}}\\
&\le \sqrt{2\pi} \exp\brac{\frac{13}{4}+\frac{20}{d}} + \exp\brac{\frac{9}{4}+\frac{12}{d}},
\end{align*}
where we use the fact that $2\sqrt{2/d} \le 1/2 + 4/d$ in the last inequality.
Consequently, applying \eqref{warmstartimplication_separation} gives us the desired bound on the average number of proposals, i.e.,
\begin{align*}
    \mathbb{E}_{\mu^k}[n_y]\leq \sqrt{2\pi} M\exp\brac{\frac{13}{4}+\frac{20}{d}} + M \exp\brac{\frac{9}{4}+\frac{12}{d}}. 
\end{align*}

As the final part of this proof, let us show that 
\begin{align*}
    \mathcal{P}_2(x)\ge \mathcal{P}_3(x), \forall x\in \R^d
\end{align*}
where $\mathcal{P}_2$ and $\mathcal{P}_3$ are respectively defined at \eqref{def:P2} and \eqref{def:P3}. The proof follows the argument showing $\mathcal{P}_1\geq \mathcal{P}_2$ in Lemma~\ref{lem:comparealltheP}. 
 Plugging in $x=y$ for \eqref{consequence:triangle}, we obtain
\begin{align}
\label{seventhone}
    -\norm{y-\proj_K(y)}-\sqrt{\frac{2\eta}{d}}\leq -\norm{y-\hat{x}}. 
\end{align}
Next, by the triangle inequality and the estimate \eqref{bound:barxkandx*}, 
\begin{align*}
    \norm{y-\proj_K(y)}\leq \norm{y-\hat x}+\norm{\hat x -\proj_K(y)}\stackrel{\eqref{bound:barxkandx*}}{\leq} \norm{y-\hat x}+\sqrt{\frac{2\eta}{d}}
\end{align*}
which further implies 
\begin{align*}
    \norm{y-\proj_K(y)}^2&\leq \norm{y-\hat x}^2+2\norm{y-\hat x}\sqrt{\frac{2\eta}{d}}+\frac{2\eta}{d}\\
    &\stackrel{\eqref{bound:barxkandx*}}{\leq} \norm{y-\hat x}^2+2\brac{\norm{y-\proj_K(y) }+\sqrt{\frac{2\eta}{d}}}\sqrt{\frac{2\eta}{d}}+\frac{2\eta}{d}. 
\end{align*}
After rearranging, this becomes 
\begin{align}
\label{fifthone}
     \norm{y-\proj_K(y)}^2-2\norm{y-\proj_K(y)}\sqrt{\frac{2\eta}{d}}-\frac{6\eta}{d}\leq \norm{y-\hat{x}}^2.
\end{align} 
In view of the definition of $\mathcal{P}_2$ in \eqref{def:P2}, combining \eqref{fifthone} and \eqref{seventhone} yields
\begin{align*}
    \mathcal{P}_2(x)&\stackrel{\eqref{def:P2}}\geq \frac{1}{2\eta}\bigg( \norm{x-\hat{x}}^2+\norm{y-\proj_K(y)}^2-2\sqrt{\frac{2\eta}{d}}\norm{y-\proj_K(y)}-\frac{6\eta}{d} \\
    &\qquad\qquad\quad-2\sqrt{\frac{2\eta}{d}}\brac{\norm{x-\hat{x}}+\norm{y-\proj_K(y)}+\sqrt{\frac{2\eta}{d}}}-\frac{12\eta}{d} \bigg)\stackrel{\eqref{def:P3}}=\mathcal{P}_3(x). 
\end{align*}
This completes the proof. 
\QEDA

\subsection{A sampling subroutine}
\label{appendix:generatesampleseparation}

While $ \exp\brac{-\frac{1}{2\eta}\left(\|x-\hat x\|^2 - 2\sqrt{\frac{2\eta}{d} } \|x-\hat x\|\right)}$ is not proportional to a Gaussian density, generating one of its samples is straightforward since it can be turned into a one-dimensional sampling problem. We state here a generic procedure for this sampling problem. An explanation is given in Lemma~\ref{lem:proposalsampling} below. 
\begin{algorithm}[H]
	\caption{Sample $X\sim \exp\brac{-\frac{1}{2\eta}\left(\|x-\hat x\|^2 - 2\sqrt{\frac{2\eta}{d} } \|x-\hat x\|\right)}$ in Algorithm \ref{alg:RGO:uniform:separation}} 
	\label{alg:separation:Psamplingforseparation}
	\begin{algorithmic}
		\State 1. Generate $W\sim {\cal N}(0,I)$  and set $\theta = W/\|W\|$; 
		\State 2. Generate $r \propto r^{d-1} \exp\left(-\frac{(r-b)^2}{2\eta}\right)$ by Adaptive Rejection Sampling for one-dimensional log-concave distribution by \cite{gilks1992adaptive}. 
		\State 3. Output $X = \hat x + r \theta$.
	\end{algorithmic}
\end{algorithm}

\begin{lemma}
\label{lem:proposalsampling}
    Algorithm \ref{alg:separation:Psamplingforseparation} generates 
    \begin{align*}
        X\sim \exp\brac{-\frac{1}{2\eta}\left(\|x-\hat x\|^2 - 2\sqrt{\frac{2\eta}{d} } \|x-\hat x\|\right)}. 
    \end{align*}

\end{lemma}

\begin{proof} By completing the square, we can see that 
\begin{align*}
    X\sim \rho(x)\propto \exp\left(-\frac{1}{2\eta}\left(\|x-\hat x\| - \sqrt{\frac{2\eta}{d}}\right)^2\right). 
\end{align*}
Let us rewrite $\rho(x)$ in polar coordinate. Set $r = \|x-\hat x\|$ and $b = \sqrt{\frac{2\eta}{d}}$. Since $\D x = r^{d-1}\D r \D\sigma(\theta)$ where $\D\sigma(\theta)$ is the surface measure of the unit sphere, we have for $r\ge 0$ and $\theta \in \mathbb{S}^{d-1}$,
\[
\rho(x)=p(r,\theta) \propto r^{d-1} \exp\left(-\frac{(r-b)^2}{2\eta}\right). 
\]
Notice the first marginal of $p$ is $p_r(r) \propto r^{d-1} \exp\left(-\frac{(r-b)^2}{2\eta}\right)$. Due to the fact that $\log p_r(r)=(d-1) \log r-\frac{(r-b)^2}{2\eta} + \mbox{const}$ and $\frac{\D^2}{\D r^2} \log p_r(r)=-\frac{d-1}{r^2}-\frac{1}{\eta}<0$, $p_r$ is a one-dimensional log-concave density.

Per the previous paragraphs, one can use any standard one-dimensional log-concave sampler, for instance \cite{gilks1992adaptive}, to sample $r\sim p_r$. Then one performs uniform sampling on the $d$-dimensional unit sphere by sampling $W\sim {\cal N}(0,I)$ and setting $\theta = W/\|W\|$ \cite{muller1959note}. Finally, one outputs $X = \hat x + r \theta$ as the sample for $p(r,\theta)=\rho(x)$. 
\end{proof}

\section{Examples of convex bodies and their oracle implementations}
\label{appen_examplesconvexbodies}

This section provides details on three types of convex bodies, namely general Z-polytopes, general H-polytopes, and bounded spectrahedron, and their corresponding oracle implementations (i.e., membership, separation, and projection). In particular, the three oracles have the same arithmetic complexity for general Z-polytopes and bounded spectrahedron.

\paragraph*{General Z-Polytopes.}

A Z-polytope in $\mathbb{R}^d$ is a polytope that can be written as a Minkowski sum of finitely many line segments, or equivalently as a linear image of a hypercube:
\[
  Z   =  \left\{ \sum_{j=1}^m t_j v_j : t_j \in [-1,1]\right\}
   =  V[-1,1]^m,\qquad V = [v_1 \dots v_m]\in\mathbb{R}^{d\times m},
\]
see, e.g., \cite{Schneider2013}.

When $m=d$ and $V$ is invertible, $Z = V[-1,1]^d$ is a parallelepiped, and the map $U \sim \mathrm{Unif}([-1,1]^d)$, $X := VU$ is a bijection with constant Jacobian determinant, so $X$ is uniform on $Z$.
Thus parallelepipeds are as easy to sample from as hypercubes.
For general Z-polytopes, $t\mapsto Vt$ is not injective and uniform $t$ on $[-1,1]^m$ does not induce uniform measure on $Z$, so uniform sampling is non-trivial and typically done via random-walk or MCMC methods e.g., \cite{Chalkis2020}.

For general Z-polytopes, membership, separation, and projection oracles can all be implemented with comparable arithmetic complexity $\mathcal{O}((m+d)^3)$ using interior-point methods.

\begin{itemize}
  \item    A membership oracle for $Z$ takes $x\in\mathbb{R}^d$ and decides whether $x\in Z$.
  Membership reduces to the box-constrained feasibility problem: find $t\in\mathbb{R}^m$ such that $Vt = x$ and $-\mathbf{1}\le t\le \mathbf{1}$. This can be formulated as a linear program (LP) and solved by interior-point methods.
  Generic primal-dual interior-point methods for dense LPs have arithmetic complexity $\mathcal{O}((m+d)^3)$; see \cite{Nesterovnemir1994}.

  \item A separation oracle for $Z$ takes $x\in\mathbb{R}^d$ and, if $x\notin Z$, returns a halfspace $\{y:a^\top y\le b\}$ containing $Z$ but not $x$.
  If the LP above is infeasible, its dual (via Farkas' lemma) yields such a separating hyperplane (see, e.g., \cite[Exercise~5.23]{Boyd2004}).
  Thus a separation oracle can be implemented with essentially the same complexity as the LP-based membership oracle: one primal-dual LP solve, again $\mathcal{O}((m+d)^3)$.

  \item A projection oracle for $Z$ maps $x\in\mathbb{R}^d$ to $\Pi_Z(x) := \arg\min_{z\in Z}\|z-x\|_2$.
  Hence projection reduces to a convex quadratic program (QP) with box constraints.
  Such QPs admit polynomial-time interior-point methods with arithmetic complexity $\mathcal{O}((m+d)^3)$ for generic dense instances \cite{Nesterovnemir1994}.
\end{itemize}

\paragraph*{General H-Polytopes.}

An $H$-polytope in $\mathbb{R}^d$ has the form
\[
  P  =  \{x\in\mathbb{R}^d : Ax \le b\},
\]
where $A\in\mathbb{R}^{m\times d}$, $b\in\mathbb{R}^m$, with componentwise inequalities.

This class includes some simple convex bodies with explicit uniform samplers: (a) the hypercube $[-1,1]^d$, sampled by drawing coordinates independently from $\mathrm{Unif}[-1,1]$; (b) the standard simplex $\Delta^{d-1} = \{x\in\mathbb{R}^d : x_i\ge 0,\ \sum_i x_i = 1\}$, sampled via independent exponentials followed by normalization; see \cite[Chapter 9]{Devroye1986}. 
These examples have special product or radial structure. For a general H-polytope with no special structure, uniform sampling usually relies on MCMC (ball walk, hit-and-run, Dikin walk, etc.; see \cite{dyer1991random,lovasz2006fast}).

For general H-polytopes, membership and separation oracles have complexity $\mathcal{O}(md)$, while projection is costlier.

\begin{itemize}
  \item   A membership oracle takes $x\in\mathbb{R}^d$ and checks $Ax \le b$.
  Computing $Ax$ costs $\mathcal{O}(md)$ and comparisons cost $\mathcal{O}(m)$, so membership runs in $\mathcal{O}(md)$ time.

  \item   A separation oracle for $P$ takes $x\in\mathbb{R}^d$ and either certifies $x\in P$ or returns a halfspace containing $P$ but not $x$.
  If $x\notin P$, some inequality is violated, say $(Ax)_i > b_i$.
  Let $a := A_{i,\cdot}$ and $\beta := b_i$.
  Then every $y\in P$ satisfies $a^\top y \le \beta$ while $a^\top x > \beta$, so the hyperplane $\{y : a^\top y = \beta\}$ separates $x$ from $P$.
  Thus separation costs $\mathcal{O}(md)$, the same order as membership.

  \item   A projection oracle for $P$ maps $x\in\mathbb{R}^d$ to $ \Pi_P(x) := \arg\min_{z\in\mathbb{R}^d} \tfrac12\|z-x\|_2^2$ s.t. $Az\le b$. This is a convex QP with linear inequality constraints.
  Generic polynomial-time algorithms via interior-point methods have worst-case arithmetic complexity $\mathcal{O}((m+d)^3)$ for dense problems \cite{Nesterovnemir1994}.
\end{itemize}

\paragraph*{Bounded Spectrahedron.}

Let $\mathbb{S}^n$ be the space of real symmetric $n \times n$ matrices with Frobenius inner product $\langle A,B\rangle = \mathrm{Tr}(A^\top B)$.
The positive semidefinite cone $\mathbb{S}^n_+ := \{ X \in \mathbb{S}^n : X \succeq 0 \}$ is closed and convex \cite[Section~2.2.4]{Boyd2004}.

For a fixed radius $R>0$, consider the bounded spectrahedron
\[
  K_R := \{ X \in \mathbb{S}^n : X \succeq 0,\ \mathrm{Tr}(X) \le R \}.
\]
It is the intersection of $\mathbb{S}^n_+$ with an affine trace constraint, hence convex.
For $X \succeq 0$ we have $\|X\|_F^2 = \sum_i \lambda_i^2 \le (\sum_i \lambda_i)^2 = \mathrm{Tr}(X)^2 \le R^2$, so $K_R$ is bounded.
Being also closed, $K_R$ is compact.
Via the identification of $\mathbb{S}^n$ with $\mathbb{R}^d$, $d = n(n+1)/2$, equipped with $\|\cdot\|_F$, $K_R$ is a compact convex body in a Euclidean space.

Membership, separation, and projection oracles on $K_R\subset\mathbb{S}^n$ can all be implemented in time $\mathcal{O}(n^3)$, dominated by eigenvalue decompositions.

\begin{itemize}
  \item 
 Membership check is performed as followed: given $X \in \mathbb{S}^n$, compute an eigendecomposition $X = Q  \mathrm{diag}(\lambda_1,\dots,\lambda_n) Q^\top$ and check $\lambda_i \ge 0$ and $\sum_{i=1}^n \lambda_i \le R$.
  The sum is $\mathcal{O}(n)$, while symmetric eigendecomposition for dense $X$ costs $\mathcal{O}(n^3)$ \cite{trefethen2022numerical}, so membership is $\mathcal{O}(n^3)$.  

  \item  Regarding the separation oracle, 
  if $X \notin K_R$, either $X$ is not PSD (some $\lambda_{\min}<0$) or $X \succeq 0$ but $\mathrm{Tr}(X) > R$.

  \emph{PSD violation:}
  Let $(\lambda_{\min},v)$ be an eigenpair with $\lambda_{\min}<0$ and $\|v\|_2=1$, set $N := v v^\top$, and define $\varphi(Y) := \langle N,Y\rangle = v^\top Y v$.
  For any $Y \succeq 0$ we have $v^\top Y v \ge 0$, so every $Y \in K_R$ satisfies $\langle N,Y\rangle \ge 0$, i.e., $K_R \subseteq \{Y : \langle N,Y\rangle \ge 0\}$.
  For $X$ we have $\langle N,X\rangle = v^\top X v = \lambda_{\min} < 0$, so the hyperplane $\{ Y : \langle N,Y\rangle = 0 \}$ strictly separates $X$ from $K_R$.

  \emph{Trace violation:}
  If $X \succeq 0$ but $\mathrm{Tr}(X) > R$, take $N := I_n$ and $\psi(Y) := \langle N,Y\rangle = \mathrm{Tr}(Y)$.
  Every $Y \in K_R$ satisfies $\mathrm{Tr}(Y) \le R$, i.e.\ $\langle N,Y\rangle \le R$, so $K_R \subseteq \{Y : \langle N,Y\rangle \le R\}$.
  For $X$ we have $\langle N,X\rangle = \mathrm{Tr}(X) > R$, so the hyperplane $\{ Y : \langle N,Y\rangle = R \}$ strictly separates $X$ from $K_R$.

  In both cases, the separator uses eigeninformation already computed for membership, so separation also costs $\mathcal{O}(n^3)$.

  \item 
  For the Frobenius-norm projection $\Pi_{K_R}(V)$, we start by performing the eigendecomposition $V = Q  \mathrm{diag}(\lambda_1,\dots,\lambda_n) Q^\top$.
  Since $K_R$ is a spectral set (invariant under orthogonal conjugation), the projection keeps eigenvectors and modifies only eigenvalues: we have $\Pi_{K_R}(V) = Q \mathrm{diag}(\mu_1,\dots,\mu_n) Q^\top$
  where $\mu \in \mathbb{R}^n$ solves $ \min_{\mu \in \mathbb{R}^n} \sum_{i=1}^n (\mu_i - \lambda_i)^2$ s.t. $ \mu_i \ge 0,\ \sum_{i=1}^n \mu_i \le R$. This is the Euclidean projection of $\lambda$ onto $C_R := \{\mu \in \mathbb{R}^n : \mu \ge 0,\ \|\mu\|_1 \le R\}$, the nonnegative $\ell_1$-ball of radius $R$.
  Projection onto $C_R$ is equivalent to projection onto an $\ell_1$-ball with nonnegativity and can be done by the same algorithms used for simplex/$\ell_1$-ball projection, e.g., \cite{Duchi2008,Condat2015}, with $\mathcal{O}(n\log n)$ or near-linear complexity.

  Thus, a projection oracle for $K_R$ is computed as follows: 1) compute $V = Q  \mathrm{diag}(\lambda) Q^\top$ (cost $\mathcal{O}(n^3)$); 2) project $\lambda$ onto $C_R$ to obtain $\mu$; and 3) return $\Pi_{K_R}(V) = Q \mathrm{diag}(\mu) Q^\top$. The overall complexity dominated by the eigendecomposition is $\mathcal{O}(n^3)$.
\end{itemize}

\end{document}